\documentclass[a4paper,UKenglish,cleveref, autoref]{lipics-v2019}

\nolinenumbers

\title{Syntactic Interpolation for Tense Logics and Bi-Intuitionistic Logic via Nested Sequents}

\titlerunning{Syntactic Interpolation via Nested Sequents}

\author{Tim Lyon}{Institut f\"ur Logic and Computation, Technische Universit\"at Wien, Austria \and \url{https://logic-cs.at/phd/students/timothy-lyon/}}{lyon@logic.at}{}{Supported by the European Union Horizon 2020 Marie Skłodowska-Curie grant No 689176 and FWF projects I 2982, Y544-N2, and W1255-N23.}

\author{Alwen Tiu}{Research School of Computer Science, The Australian National University, Australia}
{
}{}{}

\author{Rajeev Gor\'e}{Research School of Computer Science, The Australian National University, Australia}
{
}{}{}

\author{Ranald Clouston}{Research School of Computer Science, The Australian National University, Australia}
{
}{}{}

\authorrunning{T. Lyon et al.}

\Copyright{Tim Lyon and Alwen Tiu and Rajeev G\'ore and Ranald Clouston}

\ccsdesc[100]{Theory of Computation~Proof Theory}
\ccsdesc[100]{Theory of Computation~Automated Reasoning}

\keywords{Bi-intuitionistic logic, Interpolation, Nested calculi, Proof theory, Sequents, Tense logics}

\category{}

\relatedversion{}

\supplement{}

\EventEditors{Maribel Fern\'{a}ndez and Anca Muscholl}
\EventNoEds{2}
\EventLongTitle{28th EACSL Annual Conference on Computer Science Logic (CSL 2020)}
\EventShortTitle{CSL 2020}
\EventAcronym{CSL}
\EventYear{2020}
\EventDate{January 13--16, 2020}
\EventLocation{Barcelona, Spain}
\EventLogo{}
\SeriesVolume{152}
\ArticleNo{28}

\usepackage{amsmath}
\usepackage{proof,bussproofs}
\usepackage{latexsym}
\usepackage{amssymb}
\usepackage{amsmath}
\usepackage{cmll}
\usepackage{MnSymbol}
\usepackage{comment}
\usepackage[all]{xy}
\usepackage{pstricks}
\usepackage{pst-node}
\usepackage{pst-tree}
\usepackage{todonotes}

\newcommand\seq[2]{{#1} \vdash {#2}}
\newcommand\iseq[3]{{#1} \vdash {#2} \mathrel{\Vert} {#3}}

\newcommand\pseq[2]{{#1} \Vdash {#2}}
\newcommand\piseq[3]{{#1} \Vdash {#2} \mathrel{\Vert} {#3}}

\def\FBox{\square}
\def\FDia{\lozenge}
\def\PBox{\blacksquare}
\def\PDia{\blacklozenge}
\def\excl{\ensuremath{- \! \! \! <}}
\def\QDia{\langle ? \rangle}
\def\QBox{[?]}

\def\impl{\supset}
\newcommand\cseq[3]{{#1} \vdash {#2} \mathrel{\Vert} {#3}}

\def\monl{monl}
\def\monr{monr}
\def\orthrule{orth}
\def\ctr{ctr}
\def\wk{wk}
\def\cut{cut}

\def\kt{\mathsf{Kt}}
\def\biint{\mathsf{BiInt}}
\def\ktp{\mathsf{Kt\Pi}}

\def\tenseone{\mathsf{Kt\Pi L}}
\def\tensetwo{\mathsf{Kt\Pi LI}}
\def\biintone{\mathsf{BiIntL}}
\def\biinttwo{\mathsf{BiIntLI}}

\newcommand\orth[1]{({#1})^\bot}
\newcommand\trans[1]{{#1}^t}
\def\Rsym{\mathcal{R}}
\def\Isym{\mathcal{I}}

\def\trans{\text{orthogonal}}

\newtheorem{innercustomlem}{\textbf{Lemma}}
\newenvironment{customlem}[1]
  {\renewcommand\theinnercustomlem{\textbf{#1}}\innercustomlem}
  {\endinnercustomlem}
  
  \newtheorem{innercustomthm}{\textbf{Theorem}}
\newenvironment{customthm}[1]
  {\renewcommand\theinnercustomthm{\textbf{#1}}\innercustomthm}
  {\endinnercustomthm}
  
\begin{document}

\maketitle

\begin{abstract}
We provide a direct method for proving Craig interpolation for a
range of modal and intuitionistic logics, including those containing a ``converse'' modality. We demonstrate this method for classical tense logic, its extensions with path axioms, and for bi-intuitionistic logic. These logics do not have straightforward
formalisations in the traditional Gentzen-style sequent calculus, but have all been shown to have cut-free nested sequent calculi.
The proof of the interpolation theorem uses these calculi and is
purely syntactic, without resorting to embeddings, semantic arguments, or interpreted connectives external to the underlying logical language. A novel feature of our proof includes an orthogonality condition for defining duality between interpolants.
\end{abstract}

\section{Introduction}

\newcommand{\Lg}[1]{\ensuremath{\mathbf{#1}}}

The Craig interpolation property for a logic \Lg{L} states that if $A \Rightarrow B \in \Lg{L}$, then there
exists a formula $C$ in the language of \Lg{L} such that $A \Rightarrow C \in
\Lg{L}$ and
$C \Rightarrow B \in \Lg{L}$, and every propositional variable appearing in $C$ appears in $A$ and $B$.
This property has many useful applications: it can be used to prove Beth
definability~\cite{interpolation-beth-kihara-ono}; in computer-aided verification it can be
used to split a large problem involving $A \Rightarrow B$ into smaller problems involving $A \Rightarrow C$ and
$C \Rightarrow B$~\cite{DBLP:reference/mc/McMillan18}; and in knowledge representation (uniform) interpolation can be used to conceal or forget
irrelevant or confidential information in ontology querying~\cite{DBLP:conf/ijcai/LutzW11}. Therefore, demonstrating that a logic possesses the Craig interpolation property is of practical value. 

Interpolation can be proved semantically or syntactically.
In the semantic method,
\Lg{L} is the set of valid formulae, thereby requiring a semantics for
\Lg{L}. In the syntactic method, often known as Maehara's
method~\cite{maehara-method}, \Lg{L} is
the set of theorems, thereby requiring a proof-calculus. 
The syntactic approach constructs the interpolant $C$ by induction on the (usually cut-free) derivation
of $A \Rightarrow B$, and usually also provides derivations
witnessing $A \Rightarrow C$ and $C \Rightarrow B$.

Over the past forty years, Gentzen's original sequent calculus has been extended in many different ways to handle a
plethora of logics. The four main extensions are hypersequent calculi~\cite{DBLP:conf/lics/CiabattoniGT08}, display calculi~\cite{DBLP:journals/igpl/Gore98}, nested
sequent calculi~\cite{DBLP:journals/aml/Brunnler09,DBLP:journals/sLogica/Kashima94,DBLP:journals/rsl/Poggiolesi08}, and labelled calculi~\cite{DBLP:journals/jphil/Negri05}. Various interpolation results have been found using these calculi
but the only general methodology  that we know of is the recent work
of Kuznets~\cite{DBLP:journals/apal/Kuznets18} with Lellman \cite{kuznet-lellman-intermediate}. Although they use extended
sequent calculi, binary-relational Kripke semantical arguments are
crucial for their methodology, and extending their method to other
semantics is left as further work.
They also construct the interpolants using a language containing (interpreted) meta-level connectives which are external
to the logic at hand, and do not
handle logics containing converse modalities such as tense
logic. Finally,  their method
does not yield derivations witnessing $A \Rightarrow C$ and $C \Rightarrow B$.\looseness=-1

We give a general, purely syntactic, methodology for proving
Craig interpolation using nested sequent calculi 
for a variety of propositional, non-classical logics including normal tense logics,
their extensions with path axioms, and bi-intuitionistic logic.
Our methodology does not
utilise semantics, does not embed one logic into another, and does not utilise logical connectives
which are external to the underlying logical language.

The first novelty of our approach is a generalisation of the notion of interpolant from formulas to sets of sequents. 
The second is a notion of orthogonality which
gives rise to a notion of duality via cut: if two interpolants are orthogonal, then the empty sequent is derivable from the sequents in the
interpolants using only the cut and the contraction rules. This duality via cut allows us to relate our more general notion of interpolants (as sets of sequents) to
the usual notion of interpolants (as formulas).
Moreover, given a derivation of $A \Rightarrow B$, our orthogonality condition not only allows us to construct the interpolant $C$, but
also the derivations witnessing $A \Rightarrow C$ and $C \Rightarrow B$. 
This fact shows that our approach possesses a distinct complexity-theoretic advantage over the semantic approach: to verify that $C$ is indeed the interpolant of $A \Rightarrow B$, one need only check the derivations of $A \Rightarrow C$ and $C \Rightarrow B$, which is a PTIME process. In the semantic approach, to verify that $A \Rightarrow C$ and $C \Rightarrow B$ are indeed valid (and that $C$ is in fact an interpolant of $A \Rightarrow B$) one must construct proofs of the implications, which is generally much harder (e.g., finding a proof of a validity in one of the tense logics presented in Sec.~\ref{sec:classical} is PSPACE complete).


{\bf Related work.} Interpolation has been heavily investigated in the description
logic community, where it is used to hide or forget
information~\cite{DBLP:journals/jair/CateFS13}. In this setting, the
logic \textsf{ALC} is a syntactic variant of the multimodal normal modal logic \textsf{Kn}
while its extension with inverse roles, \textsf{ALCI}, is a variant of the multimodal normal tense
logic \textsf{Ktn}. 
Cate et al~\cite{DBLP:journals/jair/CateFS13} utilise a complexity-optimal tableau
algorithm to prove interpolation for \textsf{ALC} via Maehara's method. They
then embed \textsf{ALCI} into \textsf{ALC} and extend their interpolation result to
\textsf{ALCI}.

By contrast, our methodology is direct: we obtain interpolation for
the normal tense logic \textsf{Kt}, and can then extract
interpolation for the normal modal logic \textsf{K} by simply observing
that our nested sequent calculus obeys the separation property: if the
end-sequent $\vdash A \rightarrow B$ contains no occurrences of the
black (converse) modalities, then neither does the interpolant.

As mentioned earlier, the work of Kuznets et
al.~\cite{FittingK15,DBLP:journals/apal/Kuznets18,kuznet-lellman-intermediate}
on interpolation for modal logics in nested sequent calculi 
is closest to ours. Our construction of interpolants for tense logics shares some similarity with theirs. One crucial difference is that our interpolants are justified purely through syntactic and proof-theoretic means, whereas their interpolants are justified via semantic arguments.  
Another important difference is that our method extends to the bi-modal case and also (bi-)intuitionistic case, and it is straightforward to adapt our work to the multi-modal case, e.g., using nested sequent calculi as in \cite{GorIanTiu12}. 
Kowalski and Ono~\cite{KowalskiO17} showed interpolation for bi-intuitionistic logic using a sequent calculus with analytic cut.  In contrast, our proof is based on a cut-free nested sequent calculus~\cite{GorPosTiu08}.


{\bf Outline of the paper.} 
In Sec.~\ref{sec:overview} we give a brief overview of a typical interpolation proof using the traditional sequent calculus, and highlight some issues of extending it to nested sequent calculi, which motivates the generalisation of the interpolation theorem we adopt in this paper. In Sec.~\ref{sec:classical} we show how the generalised notion of interpolants can be used to prove the Craig interpolation theorem for classical tense logic and its extensions with path axioms~\cite{GorPosTiu11}, covering all logics in the modal cube and more. We then show how our approach can be extended to bi-intuitionistic logic in Sec.~\ref{sec:intuitionistic}.
In Sec.~\ref{sec:concl} we conclude and discuss future work.


\section{Overview of our approach}
\label{sec:overview}

We analyze a typical syntactic interpolation proof for
Gentzen sequents, highlight the issues of extending it to nested sequents,
and motivate our syntactic approach for interpolation.

Consider, for example, a two-sided sequent calculus for classical logic such
as~\textsf{G3c}~\cite{basic-proof-theory-troelstra}.
Interpolation holds when we can prove
that for all $\Gamma_1, \Gamma_2, \Delta_1, \Delta_2$,
if 
$\seq{\Gamma_1, \Gamma_2 }{\Delta_1,\Delta_2}$ is provable in \textsf{G3c},
then 
so are both $\seq{\Gamma_1}{\Delta_1,C}$ and $\seq{C,\Gamma_2}{\Delta_2}$,  
for some $C$ containing only propositional variables common to both $\Gamma_{1}, \Delta_{1}$ and $\Gamma_{2},\Delta_{2}$.

The inductive construction of $C$ can be encoded via
inference rules over more expressive 
sequents that specify the splitting of the contexts
and the interpolant constructed thus far.  
In~\textsf{G3c}, we write
$\cseq{\Gamma_1 \mid \Gamma_2}{\Delta_1 \mid \Delta_2}{C}$ to denote the sequent
$\seq{\Gamma_1,\Gamma_2}{\Delta_1,\Delta_2}$ with its context split into
$\Gamma_1 \vdash \Delta_1$ and $\Gamma_2\vdash \Delta_2$, and with $C$ the
interpolant.  Inference rules for this extended sequent are similar to the
usual ones, with variations encoding the different ways the
contexts may be split. For example, the initial rule $\seq{\Gamma, p}{p, \Delta}$ has the following four variants
corresponding to the four splittings of where $p$ can occur (with four different interpolants!):\looseness=-1
$$
\infer[]
{\cseq{\Gamma_1, p \mid \Gamma_2}{p, \Delta_1 \mid \Delta_2}{\bot}}
{}
\qquad
\infer[]
{\cseq{\Gamma_1, p \mid \Gamma_2}{\Delta_1 \mid p, \Delta_2}{p} }
{}
$$
$$
\infer[]
{\cseq{\Gamma_1 \mid \Gamma_2, p}{p, \Delta_1 \mid \Delta_2}{\neg p}}
{}
\qquad
\infer[]
{\cseq{\Gamma_1 \mid \Gamma_2, p}{\Delta_1 \mid p, \Delta_2}{\top}}
{}
$$
Branching rules, such as the right-introduction rule for $\land$, split into 
two variants, depending on whether the principal formula is in the first or the second partition of the context:
$$
\infer[\land_{R_1}]
{\cseq{\Gamma_1 \mid \Gamma_2}{A \land B , \Delta_1\mid \Delta_2}{C \lor D}}
{
\cseq{\Gamma_1 \mid \Gamma_2}{A, \Delta_1\mid \Delta_2}{C}
&
\cseq{\Gamma_1 \mid \Gamma_2}{B, \Delta_1 \mid \Delta_2}{D}
}
$$
$$
\infer[\land_{R_2}]
{\cseq{\Gamma_1 \mid \Gamma_2}{\Delta_1 \mid \Delta_2, A \land B}{C \land D}}
{
\cseq{\Gamma_1 \mid \Gamma_2}{\Delta_1 \mid \Delta_2, A}{C}
&
\cseq{\Gamma_1 \mid \Gamma_2}{\Delta_1 \mid \Delta_2, B}{D}
}
$$
Observe that the interpolants of the conclusion sequents are composed from the
interpolants of the premises, but with the main connectives dual to one another: a disjunction in the $\land_{R_1}$ rule and a conjunction in the $\land_{R_2}$ rule.  These observations also apply for the other rules of
\textsf{G3c}, with a slight subtlety for the implication-left rule: see~
\cite{basic-proof-theory-troelstra}. Interpolation  for
\textsf{G3c} can then be proved by a straightforward induction on the height of
proofs.

Below we discuss some issues with extending this approach to proving interpolation for modal/tense logics and bi-intuitionistic logic using nested sequent calculi, and how these issues lead to the generalisation of the intermediate lemmas we need to prove (which amounts to an interpolation theorem
for sequents, rather than formulae). 

\paragraph*{Classical modal and tense logics} 
A nested sequent~\cite{DBLP:journals/sLogica/Kashima94} can be seen as a tree of
traditional Gentzen-style sequents. For classical modal logics, single-sided
sequents suffice, so a nested sequent in this case
can be seen as a nested multiset: i.e. a multiset whose elements can be
formulae or multisets.  Following the notation in \cite{GorPosTiu11}, a sequent
nested inside another sequent is prefixed with a $\circ$, which is the
structural proxy for the $\Box$ modal operator.  For example,
the nested sequent below first left, with two sub-sequents $\{c, d\}$ and $\{e,f\}$, 
represents the formula shown second left:
\begin{center}
  \begin{math}
    \seq{}{\{a, b, \circ\{c, d\}, \circ\{e,f\} \}}
    \qquad
    a \lor b \lor \FBox(c \lor d) \lor \FBox(e \lor f)
    \qquad
    \infer[]
    {\seq{}{ \Gamma, \FDia A, \circ\{\Delta\} }}
    {\seq{}{ \Gamma, \circ\{A, \Delta\} }}
    \qquad
    \infer[]
    {\seq{}{\FDia \neg p, \FDia p, \circ\{q\}}}
    {\seq{}{\FDia \neg p, \circ \{p, q\}}}
  \end{math}
\end{center}

Nested sequent calculi for modal
logics~\cite{DBLP:journals/aml/Brunnler09,GorPosTiu11,DBLP:journals/sLogica/Kashima94}
typically contain the {\em propagation rule} for diamond shown third left above which
``propagates'' the $A$ into the scope of $\circ$, when read upwards.  
Propagation rules complicate the adaptation of the interpolation
proof from traditional Gentzen sequent calculi.  In particular, it is not
sufficient to partition a context into two disjoint multisets. That is, suppose
a nested sequent $\seq{}{\Gamma,\Delta}$ is provable, and we would like to
construct an interpolant $C$ such that $\seq{}{\Gamma, C}$ and
$\seq{}{\overline{C}, \Delta}$ are provable, where $\overline C$ is the negation normal form of $\neg C$. Suppose the proof of
$\seq{}{\Gamma, \Delta}$ ends with a propagation rule, e.g., when
$\Gamma = \FDia \neg p, \FDia p$ and $\Delta = \circ\{q\}$
as shown above far right.
In this case, by induction, we can construct an interpolant $D$ such that
the splittings $\seq{}{\FDia \neg p, D}$ and $\seq{}{ \overline D, \circ\{p,
  q\}}$ of the premiss are
provable, but
it is in general not obvious how to construct the desired interpolant $C$
for the conclusion 
$\seq{}{\FDia \neg p, \circ \{p, q\}}$
from $D$.  For this example,
$D$ should be $\FBox p$, 
$C$ should be $\FBox \bot$, which does not mention $p$ at all.

The above issue with propagation rules suggests that we need to strengthen the
induction hypothesis to construct interpolants, i.e., by considering splitting
the sequent context at every sub-sequent in the nested sequent. For example, the
nested sequent $\seq{}{\FDia \neg p, \circ \{p, q\}}$ above should be split into
$\seq{}{\FDia \neg p, \circ\{p\} }$ and $\seq{}{\circ \{ q \} }$ when applying
the induction hypothesis. Then, $D = C = \FBox \bot$ is indeed an interpolant:
both $\seq{}{\FDia \neg p, \circ\{p\}, \FBox \bot }$ and
$\seq{}{\FDia \top, \circ \{ q \} }$ are provable. Nevertheless, employing a formula interpolant is not enough
to push through the inductive argument in general. Consider, for example, the nested sequent
$\seq{}{\circ \{p, \neg p\} }$, which is provable with an identity rule, and its
partition $\seq{}{\circ \{ p \}}$ and $\seq{}{\circ\{ \neg p \}}$. There is no
formula $C$ such that both $\seq{}{\circ\{p \}, C}$ and
$\seq{}{\overline C, \circ\{\neg p\}}$ are provable.  One solution to this
problem is to generalise the interpolation statement to consider a nested
sequent as an interpolant: If a nested sequent $\seq{}{\Gamma}$ is provable,
then for every `partitioning' of $\seq{}{\Gamma}$ into $\seq{}{\Gamma_1}$ and $\seq{}{\Gamma_2}$ (where
the partitioning applies to every sub-sequent in a nested sequent; the precise
definition will be given in subsequent sections), there exists $\seq{}{\Delta}$ (the
interpolant), $\seq{}{\Gamma_1'}$ and $\seq{}{\Gamma_2'}$ such that
\begin{enumerate}
 \item The propositional variables occuring in $\seq{}{\Delta}$ are in both $\seq{}{\Gamma_1}$ and $\seq{}{\Gamma_2}$,
 \item $\seq{}{\Gamma_1'}$ splits into $\seq{}{\Gamma_1}$ and $\seq{}{\Delta}$, and $\seq{}{\Gamma_2'}$ splits into $\seq{}{\Gamma_2}$ and $\seq{}{\overline \Delta}$, 
 where $\seq{}{\overline \Delta}$ denotes the nested sequent $\seq{}{\Delta}$ with all formula occurrences replaced with their negations, and
 \item Both $\seq{}{\Gamma_1'}$ and $\seq{}{\Gamma_2'}$ are provable.
\end{enumerate}
 For example, the nested sequent $\seq{}{\circ\{p, \neg p\}}$, with partitions $\seq{}{\circ \{p\}}$ and $\seq{}{\circ\{\neg p\}}$, has the interpolant 
 $\seq{}{\Delta} = \seq{}{\circ\{\neg p\}}$ (hence $\seq{}{\overline \Delta} = \seq{}{\circ\{ p \}}$), and
 $\seq{}{\Gamma_1'} = \seq{}{\Gamma_2'} = \seq{}{\circ\{p, \neg p\}}.$

One remaining issue is that, since we now use a nested sequent as an interpolant, the composition of interpolants needs to be adjusted as well. 
Recall that in the construction of interpolants for \textsf{G3c} above, in the case involving the right-introduction for $\land$,
we constructed either $C \lor D$ or $C \land D$ as the interpolant for the conclusion.
If $C$ and $D$ are nested sequents, the expression $C \lor D$ or $C \land D$ would not be well-formed. To solve this remaining issue, we generalise the interpolant further to be a {\em set} of nested sequents. 

Fitting and Kuznets~\cite{FittingK15} similarly generalise the
notion of interpolants, but instead of generalising interpolants to a
set of (nested) sequents, they introduce `meta' connectives for conjunction and
disjunction, applicable only to interpolants, and justified
semantically. Our notion of interpolants 
requires no new logical operators or semantical notions.

\paragraph*{Propositional bi-intuionistic logic}
Bi-intuitionistic logic is obtained from intuitionistic logic by adding a subtraction (or exclusion) connective $\excl$ that is dual to implication. 
Its introduction rules are the mirror images of those for implication; in the traditional sequent calculus, these take the form:
$$
\infer[\excl_L]
{\seq{A \excl B}{\Delta}}
{\seq{A}{B, \Delta} }
\qquad
\infer[\excl_R]
{\seq{\Gamma}{\Delta, A \excl B}}
{\seq{\Gamma}{\Delta, A} & \seq{\Gamma, B}{\Delta}}
$$
However, as shown in \cite{PinUus18}, the cut rule cannot be entirely eliminated in a sequent calculus featuring these rules, although
they can be restricted to analytic cuts~\cite{KowalskiO17}. 
In \cite{GorPosTiu08}, Postniece et al. show how bi-intuitionistic logic can be formalised in a nested sequent calculus. 
Although interpolation holds for intuitionistic logic, it does not generalise straightforwardly to bi-intuitionistic logic, and only very recently 
has interpolation for bi-intuitionistic logic been shown~\cite{KowalskiO17}. 
The proof for the interpolation theorem for intuitionistic logic is very similar to the proof of the same theorem for classical logic; 
one simply needs to restrict the partitioning of the sequent to the form $\seq{\Gamma_1 \mid \Gamma_2}{ \Delta_1 \mid \Delta_2}$ where $\Delta_1$
is empty and $\Delta_2$ contains at most one formula occurrence. Since the (nested) sequent calculus for bi-intuitionistic logic uses
multiple-conclusion (nested) sequents, the proof for intuitionistic logic cannot be adapted to the bi-intutionistic case. 
The problem already shows up in the very simple case involving the identity rule: suppose we have a proof of the initial sequent $\seq{p}{p}$ and we want to partition the sequent as $\seq{\cdot \mid p}{p \mid \cdot}$. It is not possible to find an interpolant $C$ such that 
$\seq{\cdot}{p,C}$ and $\seq{C,p}{\cdot}$ (otherwise, one would be able to prove the excluded middle
$p \lor (p\supset \bot)$, which is not valid in bi-intuitionistic logic, using the cut formula $p \lor C$).
In general, the inductive construction of the interpolant for $A \impl B$ may involve finding an interpolant $C$
for the problematic partition of the form $\seq{ \cdot \mid \Gamma}{\Delta \mid \cdot}$, where $\Delta$ is non-empty. This case does not
arise in the interpolation proof for intuitionistic logic in \cite{basic-proof-theory-troelstra}, due to the restriction to single-conclusion sequents.

We show that the above issue with bi-intutionistic logic can be solved using the
same approach as in modal logic: simply extend the interpolant to
a set of nested sequents. In particular, for
$\seq{\cdot \mid p}{p \mid \cdot}$, the generalised interpolation
statement only requires finding an interpolating sequent $\seq{\Gamma}{\Delta}$
and its `dual' $\seq{\Gamma'}{\Delta'}$ (see below) such that both
$\seq{\Gamma}{p, \Delta}$ and $\seq{\Gamma',p}{\Delta'}$ are provable,
which is achieved by letting $\Gamma = \{p\}$, $\Delta = \{~\}$,
$\Gamma' = \{~\}$ and $\Delta' = \{p\}.$

\paragraph*{Interpolating sequents and orthogonality}
In a simplified form (e.g., sequent calculus), 
the generalised interpolation result we show can be roughly summarised as follows: given a provable sequent 
$\seq{\Gamma_1, \Gamma_2}{\Delta_1,\Delta_2}$, there exist two sets of sequents $\Isym$ and $\Isym'$ such that
\begin{enumerate}
\item For every sequent $(\seq{\Sigma}{\Theta}) \in \Isym$, the sequent $\seq{\Gamma_1, \Sigma}{\Delta_1, \Theta}$ is provable,
\item For every sequent $(\seq{\Sigma'}{\Theta'}) \in \Isym'$, the sequent $\seq{\Gamma_2, \Sigma'}{\Delta_2, \Theta'}$ is provable,
\item The propositional variables in $\Isym$ and $\Isym'$ occur in both $\seq{\Gamma_1}{\Delta_1}$ and $\seq{\Gamma_2}{\Delta_2}$, and 
\item The sequents in $\Isym$ and $\Isym'$ are {\em orthogonal} to each other, that is, the empty sequent $\seq{}{}$ is derivable from
all sequents in $\Isym \cup {\Isym'}$ using only the cut rule and possibly structural rules (contraction and/or weakening).
\end{enumerate}
The set $\Isym$ is taken to be the (sequent) interpolant. 

Last, the orthogonality condition, can be seen as a generalisation of duality. 
To see how this is the case, consider a degenerate case where $\seq{\Gamma_1, \Gamma_2}{\Delta_1,\Delta_2}$ is a classical sequent
(e.g., in \textsf{G3c}). We show how one can convert a formula interpolant in the usual definition (i.e., formula $C$ s.t. 
$\seq{\Gamma_1}{C, \Delta_1}$ and $\seq{\Gamma_2, C}{\Delta_2}$ are provable) to a sequent interpolant satisfying the four conditions above, and 
vice-versa.
For the forward direction, simply let $\Isym = \{ \seq{}{C}\}$
and $\Isym' = \{ \seq{C}{} \}$. It is easy to see that $\Isym$ is orthogonal to $\Isym'.$
For the converse direction, suppose we have a sequent interpolant $\Isym$ and its orthogonal $\Isym'$. We illustrate how one can
construct a formula interpolant $C$. To simplify the discussion, let us assume that 
$$
\Isym = \{ (\seq{p,q}{r,s}) \}
\quad
{\Isym'} = \{ (\seq{}{p}), ~ (\seq{}{q}), ~ (\seq{r}{}) ~ , (\seq{s}{~})\}
$$
and that the following sequents are provable:
$$(1)~\seq{\Gamma_1, p, q}{r,s, \Delta_1} \quad (2)~\seq{\Gamma_2}{\Delta_2, p} \quad (3)~\seq{\Gamma_2}{\Delta_2, q} \quad (4)~\seq{\Gamma_2,r }{\Delta_2} \quad (5)~\seq{\Gamma_2, s}{\Delta_2}$$
Let $C = (p \land q) \impl (r \lor s).$ Then it is easy to see that $\seq{\Gamma_1}{\Delta_1,C}$ is provable given (1), and
$\seq{\Gamma_2, C}{\Delta_2}$ is provable given (2) - (5).
The formal statement and the proof of the generalised interpolation theorem will be discussed in detail in the next two sections.

\paragraph*{A note on notation}
In what follows, we adopt a representation of nested sequents using restricted
labelled sequents where we use the labels and relational atoms to encode the tree structure of a nested sequent.
To clarify what we mean, consider the following nested sequent for tense logic~\cite{GorPosTiu11}:
$$
A, B, \circ\{ C, D\}, \circ\{ E, F, \bullet \{G, H\}, \circ\{ I\} \} 
$$
Graphically, the nested sequent can be represented as a tree (shown below left) with two types of edges $\overset{\circ}{\rightarrow}$ and $\overset{\bullet}{\rightarrow}$. Alternatively, the nested sequent can be represented as the polytree shown below right with a single type of edge $\overset{}{\rightarrow}$ and where the orientation of the edge encodes the two types of structures $\circ{\{}\}$ and $\bullet{\{}\}$ of the nested sequent (observe that the $\bullet$-edge from $E,F$ to $G,H$ in the left diagram has been reversed in the right diagram).\footnote{A polytree is a directed graph such that its underlying graph---the graph obtained by ignoring the orientation of the edges---is a tree.}
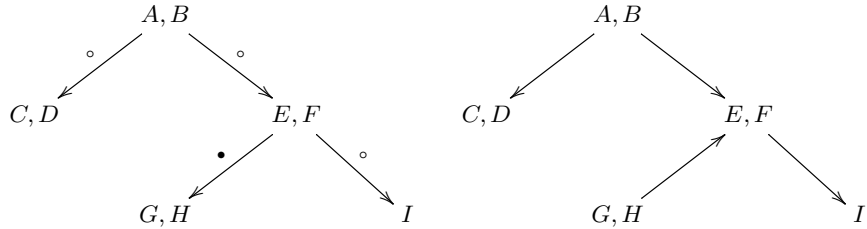
\begin{center}
\begin{footnotesize}
\begin{tabular}{c c}
\xymatrix{
 & A, B \ar@{>}[dl]_{\circ}\ar@{>}[dr]^{\circ} & & \\
 C,D & & E,F\ar@{>}[dl]_{\bullet}\ar@{>}[dr]^{\circ} & \\
  & G,H & & I
}

&

\xymatrix{
 & A, B \ar@{>}[dl]\ar@{>}[dr] & & \\
 C,D & & E,F\ar@{>}[dr] & \\
  & G,H\ar@{>}[ur] & & I
}
\end{tabular}
\end{footnotesize}
\end{center}
In the latter representation, the structure of the nested sequent can be encoded using a single binary relation: we label each node of the tree corresponding to the nested sequent (as shown above left) with unique labels $x$, $y$, $z$, $\ldots$, encode each edge $x \overset{\circ}{\rightarrow} y$ from a label $x$ to a label $y$ with a relation $Rxy$, and encode each edge $x\overset{\bullet}{\rightarrow}y$ with a relation $Ryx$~\cite{CiaLyoRam18}. The above nested sequent can then be equivalently represented as a labelled sequent
where $\Rsym = \{Ruv, Ruw, Rxw, Rwy\}$ and $R$ is a relational symbol:
$$
\seq{\Rsym}{u:A, u:B, v:C, v:D, w:E, w:F, x:G, x:H, y:I}
$$

Inference rules in a nested sequent calculus can be trivially encoded as rules
in a restricted labelled calculus seen as a `data
structure' rather than a proper labelled sequent
calculus.

We \textbf{stress} that our labelled notation to represent
nested sequents is just a matter of presentation: the labelled representation is
notationally simpler for presenting inference rules and composing nested
sequents. For instance, the operation of merging two nested sequents with
isomorphic shapes is simply the union of the multiset of labelled formulae.

\section{Interpolation for Tense Logics}
\label{sec:classical}

\begin{figure*}

$$
\overline{A} \vee \FBox \PDia A
\qquad
\overline{A} \vee \PBox \FDia A
\qquad
\FDia (A \wedge \overline{B}) \vee \FDia \overline{A} \vee \FBox B
\qquad
\PDia (A \wedge \overline{B}) \vee \PDia \overline{A} \vee \PBox B
\qquad
\infer[\FBox]
{\FBox A }
{ A }
\qquad
\infer[\PBox]
{ \PBox A }
{ A }
$$


\caption{The minimal tense logic $\kt$ consists of all classical propositional tautologies, modus ponens, and is additionally extended with the above axioms and inference rules.}
\label{fig:minimal-tense-logic-kt}
\end{figure*}

As usual, we interpret $\FBox A$ as saying that $A$ holds at every
point in the immediate future, and $\FDia A$ as saying that
$A$ holds at some point in the immediate future. Conversely, the $\PBox$ and
$\PDia$ modalities make reference to the past: $\PBox A$ says that $A$ holds at every point in the immediate past, and $\PDia A$ 
says that $A$ holds at some point in the immediate past. Last,
we take $\overline{p}$ to be the negation of $p$, and use the notation  $\QBox \in \{\FBox, \PBox\}$ and $\QDia \in \{\FDia,\PDia\}$.


We consider tense formulae in negation normal form (nnf) as
this simplifies our calculi while retaining the expressivity of the original
language. The language for the tense logics we consider is given via the following BNF grammar: 
$$
A ::= p \ | \ \overline{p} \ | \ (A \wedge A) \ | \ (A \vee A) \ | \ (\FBox A) \ | \ (\FDia A) \ | \ (\PBox A) \ | \ (\PDia A).
$$
Since our language excludes an explicit connective for negation, we define it formally below (Def.~\ref{def:tense-negation}). Using the definition, we may define an implication $A \rightarrow B$ to be $\overline{A} \vee B$.
\begin{definition}
  \label{def:tense-negation} For a formula $A$, we define the \emph{negation}
  $\overline{A}$ recursively on the structure of $A$: if $A=p$ then
  $\overline{A} := \overline{p}$ and if $A = \overline{p}$ then
  $\overline{A} :=p$. The clauses concerning the connectives are as follows: (1)
  $\overline{B \land C} := \overline{B} \vee \overline{C}$, (2)
  $\overline{\QBox B} := \QDia \overline{B}$, (3)
  $\overline{B \lor C} := \overline{B} \wedge \overline{C}$, and (4)
  $\overline{\QDia B} := \QBox \overline{B}$.
\end{definition}

Path axioms are of the form $\QBox_1 \QBox_2 \cdots \QBox_n \;\bar p \;\vee \QDia p$ (or, equivalently, $\QDia_1 \cdots \QDia_n p \to \QDia p$) with $n \in \mathbb{N}$. 
See~\cite{GorIanTiu12} for an overview of path axioms. 

The tense logics we consider are all extentions of the minimal tense logic $\kt$
(Fig.~\ref{fig:minimal-tense-logic-kt}) with path
axioms. 
Thus, $\ktp$ is the minimal extension
of $\kt$ with all axioms from the
finite set $\mathsf{\Pi}$ of path axioms.



The calculus for $\kt$, extended with a set of path axioms $\mathsf{\Pi}$, is given in Fig.~\ref{fig:labelled-rules}. Labelled sequents are defined to be syntactic objects of the form $\seq{\Rsym}{\Gamma}$, where $\Rsym$ is a multiset of relational atoms of the form $Rxy$ and $\Gamma$ is a multiset of labelled formulae of the form $x:A$, with $A$ a tense formula and labels from a countable set $\{x,y,z, \ldots\}$. 



Note that the side conditions $x \Rsym^\mathsf{\Pi} y$ and $y \Rsym^\mathsf{\Pi} x$ of the $\FDia$ and $\PDia$ rules, respectively, depend on the set $\mathsf{\Pi}$ of path axioms added to $\kt$. The definition of the relation $\Rsym^{\mathsf{\Pi}}$ is founded upon
various auxiliary concepts that fall outside the main scope of this
paper. We therefore refer the interested reader to App.~\ref{app:path-axioms-relation} 
where the $\Rsym^{\mathsf{\Pi}}$ relation as
well as the concepts needed for its definition are explicitly
provided. See also~\cite{GorPosTiu11,GorIanTiu12} for details.





\begin{lemma}\label{lm:labelled-calc-properties} The contraction rules $\ctr$, the weakening rules $\wk$ and $\cut_{1}$ are admissible, and all inference rules are invertible in $\tenseone$. 
\end{lemma}

\begin{proof} See Fig.~\ref{fig:admissible-rules-biint} for rules\footnote{Note that since $\tenseone$ uses one sided sequents, we only consider instances of the rules where labelled formulae occur solely on the right of the sequent arrow.} and \cite{GorPosTiu11,GorIanTiu12} for details.
\end{proof}

\begin{figure}
$$
\infer[id]
{\seq{\Rsym}{x:\bar p, x: p, \Delta}}
{}
\qquad
\infer[\lor]
{\seq{\Rsym}{ x:A\lor B, \Delta}}
{\seq{\Rsym}{x:A, x:B,  \Delta}}
\qquad
\infer[\land]
{\seq{\Rsym}{x:A \land B, \Delta}}
{\seq{\Rsym}{x:A, \Delta}  & \seq{\Rsym}{x:B, \Delta}}
$$
$$
\infer[\FDia, ~ x \Rsym^\mathsf{\Pi} y]
{\seq{\Rsym}{x:\FDia A, \Delta} }
{
\seq{\Rsym}{x:\FDia A, y:A, \Delta}
}
\qquad \qquad
\infer[\FBox, \mbox{ $y$ fresh}]
{\seq{\Rsym}{x:\FBox A, \Delta} }
{\seq{\Rsym, Rxy}{y:A, \Delta} }
$$
$$
\infer[\PDia, ~ y \Rsym^\mathsf{\Pi} x]
{\seq{\Rsym}{x:\PDia A, \Delta} }
{
\seq{\Rsym}{x:\PDia A, y:A, \Delta}
}
\qquad \qquad
\infer[\PBox, \mbox{ $y$ fresh}]
{\seq{\Rsym}{x:\PBox A, \Delta} }
{\seq{\Rsym, Ryx}{y:A, \Delta} }
$$
\caption{The calculus $\tenseone$ for $\kt$ extended with a set of path axioms $\mathsf{\Pi}$.} 
\label{fig:labelled-rules}
\end{figure}

\begin{example}
\label{ex:tense-proof}
Consider the formula $\FBox \FDia \overline{q} \rightarrow \FBox(\FDia \overline{p} \lor \FDia \FDia p)$, which is a theorem in the logic $\ktp$ with $\mathsf{\Pi} = \{\FBox \FBox \overline{p} \vee \FDia p\}$. A proof of this formula is provided in Fig. \ref{fig:ex-tense-proof}.
\end{example}

\begin{figure*}
{\small
$$
\infer[\lor]
{\seq{}{x:\FDia \FBox q \lor \FBox(\FDia \overline{p} \lor \FDia \FDia p)}}
{
\infer[\FBox]
{\seq{}{x:\FDia \FBox q, x: \FBox(\FDia \overline{p} \lor \FDia \FDia p)}}
 {
 \infer[\lor]
 {
 \seq{Rxy}{x: \FDia \FBox q, y: \FDia \overline{p} \lor \FDia \FDia p}
 }
  {
  \infer[\FDia]
  {\seq{Rxy}{x: \FDia \FBox q, y: \FDia \overline{p}, y: \FDia \FDia p}}
   {
   \infer[\FBox]
   {\seq{Rxy}{x: \FDia \FBox q, y: \FBox q, y: \FDia \overline{p}, y: \FDia \FDia p}}
    {
    \infer[\FDia]
    {\seq{Rxy,Ryz}{x: \FDia \FBox q, z: q, y: \FDia \overline{p}, y: \FDia \FDia p}}
     {
     \infer[\FBox]
     {\seq{Rxy,Ryz}{x: \FDia \FBox q, z: \FBox q, z: q, y: \FDia \overline{p}, y: \FDia \FDia p}}
      {
      \infer[\FDia]
      {\seq{Rxy,Ryz,Rzw}{x: \FDia \FBox q, w:q,  z: q, y: \FDia \overline{p}, y: \FDia \FDia p}}
       {
       \infer[\FDia]
       {\seq{Rxy,Ryz,Rzw}{x: \FDia \FBox q, w:q,  z: q, y: \FDia \overline{p}, w: \overline{p}, y: \FDia \FDia p}}
        {
        \infer[\FDia]
        {\seq{Rxy,Ryz,Rzw}{x: \FDia \FBox q, w:q,  z: q, y: \FDia \overline{p}, w: \overline{p}, y: \FDia \FDia p, z:\FDia p}}
         {\infer[id]
         {\seq{Rxy,Ryz,Rzw}{x: \FDia \FBox q, w:q,  z: q, y: \FDia \overline{p}, w: \overline{p}, y: \FDia \FDia p, z:\FDia p, w:p}}
          {}
         }
        }
       }
      }
     }
    }
   }
  }
 }
}
$$
}
\caption{A proof in $\tenseone$ where $\mathsf{\Pi} = \{\FBox \FBox \overline{p} \vee \FDia p\}$.}
\label{fig:ex-tense-proof}
\end{figure*}

As stated in Sec.~\ref{sec:overview}, we extend the notion of an interpolant to a set of nested sequents. 
In our definition of interpolants, we are interested only in duality via
cut. In particular, the relational atoms (encoding the tree shape of a nested
sequent) are not explicitly represented in the interpolants since they can be
recovered from the contexts of the sequents in which the interpolants are used. We therefore define 
a {\em flat sequent} to be  a sequent without relational atoms. 
For classical tense logic, a flat sequent is thus a multiset of labelled formulas.

\begin{definition}
\label{def:interpolant}
An interpolant, denoted $\Isym$, is a set of 
flat sequents.
\end{definition}

For example, the set below is an interpolant in our context:
$$\{ (\seq{}{x:A, y: B, z:W}) , (\seq{}{x:C, y:D}), (\seq{}{u: E}) \}.$$

Since our interpolant is no longer a formula, we need to define the dual of an interpolant in order to generalise the statement of 
the interpolation result to sequents.  We have informally explained in Sec.~\ref{sec:overview}
that duality in this case is defined via cut.
%
Intuitively, given an interpolant $\Isym_1$, 
its \emph{dual} is any set of nested sequents $\Isym_2$ such that the empty sequent can be derived from $\Isym_1$ and $\Isym_2$ using cut (possibly with contraction). 
For example, given $\Isym_1 = \{ (\seq{}{x:A, y:B}), (\seq{}{u:C}), (\seq{}{v:D}) \}$, 
there are several candidates for its dual:
$$
\begin{array}{rcl}
\Isym_2 & = & \{ (\seq{}{x:\overline A }), ~ (\seq{}{y:\overline B, u:\overline C, v:\overline D})  \} \\
\Isym_3 & = & \{ (\seq{}{x:\overline A, u:\overline C, v:\overline D}),~ (\seq{}{y:\overline B}) \} \\
\Isym_4 & = & \{ (\seq{}{x: \overline A, u:\overline C}), (\seq{}{x:\overline A, v:\overline D}), (\seq{}{y:\overline B, u:\overline C}), 
\{y:\overline B, v:\overline D\} \}
\end{array}
$$
The empty sequent can be derived from 
$\Isym_1 \cup {\Isym_i}$, for $i = 2,3,4$ using cut (and contraction, in the case of $\Isym_4$). 
In principle, any of the dual candidates to $\Isym_1$ can be used, but to make the construction
of the interpolants deterministic, our definition below will always choose $\Isym_4$, as it is relatively
straightforward to define as a function of $\Isym_1.$

\begin{definition}
\label{def:orthogonal}
For an interpolant $\Isym = \{ \seq{}{ \Lambda_1} , \ldots , \seq{}{\Lambda_n} \}$,
the {\em $\trans$} $\orth{\Isym}$ is defined as
$$
\orth{\Isym} = \{ (\seq{}{x_1:\overline{A_1}, \ldots, x_n:\overline{A_n}})  \mid \forall i \in \{1,\dots,n\}, x_i:A_i \in \Lambda_i \}.
$$
\end{definition}
For example, the $\trans$ of $\Isym = \{(\seq{}{x:A,y:B}), (\seq{}{x:C,z:D})\}$ is
$$
\orth{\Isym} = \{(\seq{}{x:\overline{A},x:\overline{C}}), (\seq{}{x:\overline{A},z:\overline{D}}), 
(\seq{}{y:\overline{B},x:\overline{C}}), (\seq{}{y:\overline{B},z:\overline{D}})\}.
$$

\begin{definition}
\label{def:box-dia-interpolants}
Let $\Isym$ be the interpolant
$$
\{ (\seq{}{ \Delta_1, y:B_{1,1}, \ldots, y:B_{1,k_1}}) , \ldots ,
   (\seq{}{ \Delta_n, y:B_{n,1}, \ldots, y:B_{n,k_n} }) \}
$$
where $y$ does not occur in $\Delta_1, \ldots, \Delta_n$ 
and define
$$
\QBox \Isym^y_x := \{ (\seq{}{ \Delta_1, x: \QBox \bigvee_{j=1}^{k_1} B_{1,j} }) , 
  \ldots,
  (\seq{}{ \Delta_n, x:\QBox \bigvee_{j=1}^{k_n} B_{n,j} }) \}
$$
where an empty disjunction is $\bot$. 
\end{definition}

\begin{definition}
\label{def:interpolant-sequent}
We define an \emph{interpolation sequent} to be a syntactic object of the form 
$\iseq{\Rsym, \Gamma_1 \mid \Gamma_2}{\Delta_1 \mid \Delta_2}{\Isym}$, 
where $\Rsym$ is a set of relational atoms, $\Gamma_{i}$ and $\Delta_{i}$ are multisets of labelled formulae (for $i \in \{1,2\}$), and $\Isym$ is an interpolant. Note that in the interpolation calculus $\tensetwo$, $\Gamma_{1} = \Gamma_{2} = \emptyset$ (see Fig.~\ref{fig:int-rules}).
\end{definition}

The vertical bar $\mid$ in an interpolation sequent marks where the sequent will be partitioned, with the left partition
serving as the antecedent and the right partition serving as the consequent in the interpolation statement. 
For example, the initial interpolation sequent shown below left 
splits into the two sequents shown below right 
$$
\infer[id]
{
\Rsym \vdash \Gamma \mid x:p, x:\overline{p}, \Delta \mathrel{\Vert} \{(\vdash x:\top)\}
}
{}
\qquad
(\Rsym \vdash \Gamma, x:\top)
\qquad
(\Rsym \vdash x:\bot,  x:p, x:\overline{p},\Delta)
$$
where the first member $\Gamma$ of the split is placed in the left sequent and the second member $x:p,x:\overline{p},\Delta$ is placed in the right sequent (note that the relational atoms $\Rsym$ are inherited by both sequents). We think of the interpolant $x:\top$ as being \emph{implied by} the left sequent, and so, we place it in the left sequent, and we think of the interpolant as \emph{implying} the right sequent, so we place its negation (viz. $x:\bot$) in the right sequent. Observe that an application of $\cut_{1}$ between the two sequents, yields $\Rsym \vdash \Gamma, x:\overline{p}, x:p, \Delta$ without the interpolant. Performing a $\cut_{1}$ in this way \emph{syntactically establishes} (without evoking the semantics) that the interpolant is indeed an interpolant (so long as the interpolant satisfies certain other properties; cf. Lem.~\ref{lm:interpolant-labels-literals} below).

The interpolation calculus $\tensetwo$ (Fig.~\ref{fig:int-rules}) uses interpolation sequents. More importantly, the calculus succinctly represents our algorithm for constructing interpolants. Most of the rules are straightforward counterparts of the 
proof system in Fig.~\ref{fig:labelled-rules}, except for the orthogonality rule $\orthrule.$ The orthogonality rule is arguably 
the most novel aspect of our interpolation calculus, as it imposes a strong requirement on our generalised  notion of
interpolants, that it must respect the underlying duality in the logic.  The key to the correctness of this rule is given in the Persistence Lemma below, which shows that double-orthogonal transformation always retains some sequents in the original interpolant. 

\begin{lemma}[Persistence]
\label{lm:double-orthogonal}
If $\seq{}{\Lambda} \in \orth{\orth{\Isym}}$, then there exists a $\seq{}{\Lambda'} \in \Isym$ such that $\Lambda' \subseteq \Lambda.$
\end{lemma}
\begin{proof}
Suppose otherwise, i.e., there exists $\seq{}{\Lambda} \in \orth{\orth{\Isym}}$ such that for all $\seq{}{\Lambda'} \in \Isym$,
we have $\Lambda' \not \subseteq \Lambda.$
Suppose $\Isym = \{\seq{}{\Lambda_1}, \ldots, \seq{}{\Lambda_n}\}.$ 
Then for each $i$, there must be a labelled formula $x_i : A_i \in \Lambda_i$
such that $x_i:A_i \not \in \Lambda.$ Let $\Theta = \{x_1:A_1, \ldots, x_n:A_n\}.$ By construction, we must have
that $\Theta \cap \Lambda = \emptyset.$
However, by Def.~\ref{def:orthogonal}, we have $\overline{\Theta} \in \orth{\Isym}$, and since $\Lambda \in \orth{\orth{\Isym}}$, 
by Def.~\ref{def:orthogonal}, $\Theta \cap \Lambda \not = \emptyset$. Contradiction.
\end{proof}

\begin{figure*}
{\small
$$
\infer[id]
{\iseq{\Rsym}{\Gamma, x:\bar p \mid x: {p}, \Delta}{\{ (\seq{}{x: p}) \} } }
{}
\quad
\infer[id]
{\iseq{\Rsym}{ \Gamma \mid x:\bar p, x:p, \Delta}{\{ (\seq{}{ x:\top})  \}}}
{}
\quad\infer[\orthrule]
{\iseq{\Rsym}{\Delta \mid \Gamma}{\orth{\Isym}}}
{\iseq{\Rsym}{\Gamma \mid \Delta}{\Isym}}
$$
$$
\infer[\lor]
{\iseq{\Rsym}{\Gamma \mid x:A\lor B,\Delta}{\Isym}}
{\iseq{\Rsym}{\Gamma \mid x:A, x:B, \Delta}{\Isym}}
\qquad
\infer[\land]
{\iseq{\Rsym}{\Gamma \mid x:A \land B, \Delta}{\Isym_1 \cup \Isym_2}}
{\iseq{\Rsym}{ \Gamma \mid x:A, \Delta}{\Isym_1}  & \iseq{\Rsym}{ \Gamma \mid x:B, \Delta}{\Isym_2}}
$$
$$
\infer[\FDia, ~ x \Rsym^\mathsf{\Pi} y]
{\iseq{\Rsym}{\Gamma \mid x:\FDia A, \Delta}{\Isym}}
{
\iseq{\Rsym}{ \Gamma \mid x:\FDia A, y:A, \Delta}{\Isym}
}
\qquad
\infer[\FBox, \mbox{ $y$ fresh}]
{\iseq{\Rsym}{ \Gamma \mid x:\FBox A, \Delta}{\FBox \Isym^y_x} }
{\iseq{\Rsym, Rxy}{ \Gamma \mid y:A, \Delta}{\Isym}}
$$
$$
\infer[\PDia, ~ y \Rsym^\mathsf{\Pi} x]
{\iseq{\Rsym}{\Gamma \mid x:\PDia A, \Delta}{\Isym}}
{
\iseq{\Rsym}{ \Gamma \mid x:\PDia A, y:A, \Delta}{\Isym}
}
\qquad
\infer[\PBox, \mbox{ $y$ fresh}]
{\iseq{\Rsym}{ \Gamma \mid x:\PBox A, \Delta}{\PBox \Isym^y_x} }
{\iseq{\Rsym, Ryx}{ \Gamma \mid y:A, \Delta}{\Isym}}
$$
}
\caption{Calculus $\tensetwo$ for constructing interpolants for $\kt$ extended with path axioms $\mathsf{\Pi}$.}
\label{fig:int-rules}
\end{figure*}

Given a formula $A$, we define the set of propositional variables $var(A)$ of $A$ to be the set $\{p \ | \ p \text{ or } \overline{p} \text{ in } A\}$. 
This notation extends straightforwardly to sets of formulae and interpolants. 

We write $\pseq{\Rsym}{\Gamma, \Delta}$ to denote that the sequent $\seq{\Rsym}{\Gamma, \Delta}$ is provable in $\tenseone$. Similarly, $\piseq{\Rsym}{\Gamma \mid \Delta}{ \Isym}$
denotes that the sequent $\iseq{\Rsym}{\Gamma \mid \Delta}{\Isym}$
is provable in $\tensetwo.$ 

\begin{definition}
\label{def:interpolation-property}
A logic has the Craig interpolation property iff for every implication $A \Rightarrow B$ in the logic, there is a formula $C$ such that (i) $var(C) \subseteq var(A) \cap var(B)$ and (ii) $A \Rightarrow C$ and $C \Rightarrow B$ are in the logic, where $\Rightarrow$ is taken to be the implication connective of the logic.
\end{definition}

We now establish that each tense logic $\ktp$ possess the Craig interpolation property when the implication connective is taken to be $\rightarrow$. 
To achieve this, we begin by showing that an interpolant sequent can be constructed from any cut-free proof.
\begin{lemma}
\label{lm:interpolant-labels-literals}
If $\pseq{\Rsym}{\Gamma, \Delta}$, then there exists an $\Isym$ such that $\piseq{\Rsym}{\Gamma \mid \Delta}{\Isym}$,  
$var(\Isym) \subseteq var(\Gamma) \cap var(\Delta)$, and all labels occuring in $\Isym$ also occur in $\Rsym,\Gamma$ or $\Delta$.
\end{lemma}
\begin{proof}
Induction on the height of the proof of $\seq{\Rsym}{\Gamma, \Delta}$ and by using the rules of $\tensetwo$.
\end{proof}

The next lemma establishes the correctness of the interpolants constructed from our interpolation calculus in Fig.~\ref{fig:int-rules}.
Its proof can be found in App.~\ref{app-proofs}. 
\begin{lemma}
\label{lm:interpolant}
For all $\Rsym, \Gamma, \Delta$ and $\Isym$, 
if $\piseq{\Rsym}{ \Gamma \mid \Delta}{\Isym}$, then 
\begin{enumerate}
\item For all $(\seq{}{\Lambda}) \in \Isym,$ we have $\pseq{\Rsym}{ \Gamma, \Lambda}$ and 
\item For all $(\seq{}{\Theta}) \in \orth{\Isym}$, we have $\pseq{\Rsym}{ \Theta, \Delta}.$
\end{enumerate}
\end{lemma}

To prove Craig interpolation, we need to construct formula interpolants. Lem.~\ref{lm:interpolant} provides
sequent interpolants, so the next step is to show how one can derive a formula interpolant from a sequent interpolant.
This is possible if the formulas in an interpolant are all prefixed with the same label.
In that case, there is a straightforward interpretation of the interpolant as a formula. 
More precisely, let $\Isym = \{(\seq{}{\Lambda_{1}}), \ldots, (\seq{}{\Lambda_{n}})\}$, where $\Lambda_{i} = \{x:A_{i,1}, \ldots, x:A_{i,k_{i}}\}$ for all $1 \leq i \leq n$. Then, its formula interpretation is given by  $\bigwedge_{i=1}^{n} \bigvee_{j=1}^{k_{i}} A_{i,j}.$
Given such an interpolant $\Isym$, we write $\bigwedge \bigvee \Isym$ to denote its formula interpretation.
The following lemma is a straightforward consequence of this interpretation.
\begin{lemma}
\label{lm:deriving-interpolant-formula}
Let $\Isym = \{\seq{}{\Lambda_{1}}, \ldots, \seq{}{\Lambda_{n}}\}$ be an interpolant with $\Lambda_{i} = \{x:A_{i,1}, \ldots, x:A_{i,k_{i}}\}$ for each $1 \leq i \leq n$. 
For any multiset of relational atoms $\Rsym$ and multiset of labelled formulae $\Gamma$, if $\pseq{\Rsym}{\Gamma,\Lambda}$ for all $\seq{}{\Lambda} \in {\Isym}$, then $\pseq{\Rsym}{\Gamma,x:\bigwedge \bigvee \Isym}$.
\end{lemma}

However, the formula-interpolant derived in Lem.~\ref{lm:deriving-interpolant-formula} gives only one-half of the full picture, as one
still needs to show that the orthogonal of a sequent interpolant admits a dual interpretation as a formula.
A key to this is the following Duality Lemma that shows that orthogonality behaves like negation. 
\begin{lemma}[Duality]
\label{lm:dual}
Given an interpolant $\Isym$, the empty sequent is derivable from $\Isym \cup \orth{\Isym}$ using the $\cut_{1}$ rule and
the contraction rule.
\end{lemma}

An interesting consequence of Duality Lemma is that it translates into duality in the above formula
interpretation as well, as made precise in the following lemma.
\begin{lemma}
\label{lm:deriving-dual-interpolant-formula}
Let $\Isym = \{\seq{}{\Lambda_{1}}, \ldots, \seq{}{\Lambda_{n}}\}$ be an interpolant with $\Lambda_{i} = \{x:A_{i,1}, \ldots, x:A_{i,k_{i}}\}$ for each $1 \leq i \leq n$. 
For any multiset of relational atoms $\Rsym$ and multiset of labelled formulae $\Delta$, if $\pseq{\Rsym}{{\Theta},\Delta}$ for all $\seq{}{\Theta} \in \orth{\Isym}$, then $\pseq{\Rsym}{x:\overline{\bigwedge \bigvee \Isym},\Delta}$.
\end{lemma}
\begin{proof}
Suppose $\orth{\Isym} = \{ (\seq{}{\Theta_1}), \ldots, (\seq{}{\Theta_k}) \}$ for some $k.$
By Lem.~\ref{lm:dual}, we have a derivation $\Xi_1$ of the empty sequent from assumptions $\Isym \cup \orth{\Isym}$. 
$$
\deduce{\seq{}{}}{\deduce{\vdots}{\seq{}{\Lambda_1} \quad \cdots \seq{}{\Lambda_n}
\qquad \seq{}{\Theta_1} \quad \cdots \quad \seq{}{\Theta_k} }}
$$
Due to admissibility of weakening (Lem.~\ref{lm:labelled-calc-properties}), for each $\Lambda_i$, there is a proof $\Psi_i$ of the sequent 
$\seq{\Rsym}{\Lambda_i,  x:\overline{\bigwedge \bigvee \Isym}}$. Adding $x:\overline{\bigwedge \bigvee \Isym}$ to every leaf sequent in $\Xi_1$ belonging to $\Isym$ gives us a derivation 
$\Xi_2$: 
$$
\infer[ctr^*]
{\seq{\Rsym}{x:\overline{\bigwedge \bigvee \Isym}}}
{
\deduce
{\seq{\Rsym}{(x:\overline{\bigwedge \bigvee \Isym})^*}}
{\deduce{\vdots}{[\seq{\Rsym}{x:F, \Lambda_1}] \quad  \cdots \quad [\seq{\Rsym}{x:F, \Lambda_n}] \qquad 
\seq{\Rsym}{\Theta_1}  \quad \cdots \quad  \seq{\Rsym}{\Theta_k}} }
}
$$
where $F = \overline{\bigwedge \bigvee \Isym}$ and sequents in brackets are provable, and where $*$ denotes multiple
copies of sequents or rules. 
By the assumption we know that each $\seq{\Rsym}{\Theta_i, \Delta}$ is provable, so by adding $\Delta$ to each
premise sequent in $\Xi_2$, we get the following proof:
$$
\infer[ctr^*]
{\seq{\Rsym}{x:\overline{\bigwedge \bigvee \Isym}, \Delta}}
{
\deduce
{\seq{\Rsym}{(x:\overline{\bigwedge \bigvee \Isym})^*, \Delta^*}}
{\deduce{\vdots}{[\seq{\Rsym}{x:F, \Lambda_1,\Delta}] \quad  \cdots \quad [\seq{\Rsym}{x:F, \Lambda_n,\Delta}] \qquad 
[\seq{\Rsym}{\Theta_1, \Delta}]  \quad \cdots \quad  [\seq{\Rsym}{\Theta_k, \Delta}]} }
}
$$
\end{proof}

\begin{theorem}
\label{tm:tense-proof-theoretic-interpolation}
If $\pseq{}{x:A \rightarrow B}$, then there exists a $C$ such that (i) $var(C) \subseteq var(A) \cap var(B)$ and (ii) $\pseq{}{x:A \rightarrow C}$ and $\pseq{}{x:C \rightarrow B}$.
\end{theorem}

\begin{corollary}
\label{cr:tense-logics-interpolation}
Every extension of the (minimal) tense logic $\kt$ with a set $\mathsf{\Pi}$ of path axioms has the Craig interpolation property.
\end{corollary}

\begin{example}
\label{ex:tense-interpolant}
Consider the formula given in example \ref{ex:tense-proof}. By making use of its derivation in Fig.~\ref{fig:ex-tense-proof}, we can apply our interpolation algorithm as shown in Fig.~\ref{fig:ex-tense-interpolant} to construct an interpolant for the formula.
\end{example}

\begin{figure*}
$$
\infer[\FBox]
{\iseq{}{x:\FDia \FBox q \mid x: \FBox(\FDia \overline{p} \lor \FDia \FDia p)}{\{(\seq{}{x:\FBox \FDia \FDia \top})\} }}
 {
 \infer[\lor]
 {
 \iseq{Rxy}{x: \FDia \FBox q \mid y: \FDia \overline{p} \lor \FDia \FDia p}{\{(\seq{}{y:\FDia \FDia \top })\}}
 }
  {
 \infer[\orthrule]
 {\iseq{Rxy}{  x: \FDia \FBox q \mid y: \FDia \overline{p}, y: \FDia \FDia p }{\{(\seq{}{y:\FDia \FDia \top })\} }}
 {
  \infer[\FDia]
  {\iseq{Rxy}{ y: \FDia \overline{p}, y: \FDia \FDia p \mid x: \FDia \FBox q }{\{(\seq{}{y:\FBox \FBox \bot })\} }}
   {
   \infer[\FBox]
   {\iseq{Rxy}{ y: \FDia \overline{p}, y: \FDia \FDia p \mid x: \FDia \FBox q, y: \FBox q}\{(\seq{}{y:\FBox \FBox \bot})\}}
    {
    \infer[\FDia]
    {\iseq{Rxy,Ryz}{ y: \FDia \overline{p}, y: \FDia \FDia p \mid x: \FDia \FBox q, z: q }\{(\seq{}{z:\FBox \bot})\}}
     {
     \infer[\FBox]
     {\iseq{Rxy,Ryz}{ y: \FDia \overline{p}, y: \FDia \FDia p \mid x: \FDia \FBox q, z: \FBox q, z: q }\{(\seq{}{z:\FBox \bot)}\}}
      {
      \infer[\orthrule]
      {\iseq{Rxy,Ryz,Rzw}{ y: \FDia \overline{p}, y: \FDia \FDia p \mid x: \FDia \FBox q, w:q,  z: q}\{(\seq{}{w:\bot})\}}
      {
      \infer[\FDia] 
      {\iseq{Rxy,Ryz,Rzw}{x: \FDia \FBox q, w:q,  z: q \mid y: \FDia \overline{p}, y: \FDia \FDia p}\{(\seq{}{w:\top})\}}
       {
       \infer[\FDia] 
       {\iseq{Rxy,Ryz,Rzw}{x: \FDia \FBox q, w:q,  z: q \mid y: \FDia \overline{p}, w: \overline{p}, y: \FDia \FDia p}\{(\seq{}{w:\top})\}}
        {
        \infer[\FDia] 
        {\iseq{Rxy,Ryz,Rzw}{x: \FDia \FBox q, w:q,  z: q \mid y: \FDia \overline{p}, w: \overline{p}, y: \FDia \FDia p, z:\FDia p}\{(\seq{}{w:\top})\}}
         {\infer[id] 
         {\iseq{Rxy,Ryz,Rzw}{x: \FDia \FBox q, w:q,  z: q \mid y: \FDia \overline{p}, w: \overline{p}, y: \FDia \FDia p, z:\FDia p, w:p}\{(\seq{}{w:\top}) \}}
          {}
         } 
        } 
       } 
      } 
      }
     }
    }
    } 
   }
  }
 }
$$
\caption{An example of the construction of tense interpolants.}
\label{fig:ex-tense-interpolant}
\end{figure*}


\section{Interpolation for Bi-Intuitionistic Logic}
\label{sec:intuitionistic}

The language for bi-intuitionistic logic $\biint$ is given via the following BNF grammar:
$$A ::= p \ | \ \top \ | \ \bot \ | \ (A \wedge A) \ | \ (A \vee A) \ | \ (A \supset A) \ | \ (A \excl A)$$

For an axiomatic definition of $\biint$ consult \cite{Rau80} and for a semantic definition see \cite{GorPosTiu08,PinUus18}.

The calculus $\biintone$ for $\biint$ is given in Fig.~\ref{fig:BiIntL-labelled-rules}. The calculus makes use of sequents of the form $\seq{\Rsym, \Gamma}{\Delta}$ with $\Rsym$ a multiset of relational atoms of the form $Rxy$, $\Gamma$ and $\Delta$ multisets of labelled formulae of the form $x : A$
(where $A$ is a bi-intuitionistic formula), and all labels are among a countable set $\{x, y, z, \ldots\}$. Note that we need not restrict the consequent of sequents to at most one formula on the right or left due to the eigenvariable condition imposed on the $\supset_{R}$ and $\excl_{L}$ rules. Moreover, for a multiset $\Rsym$ of relational atoms or a multiset $\Gamma$ of labelled formulae, we use the notation $\Rsym[x/y]$ and $\Gamma[x/y]$ to represent the multiset obtained by replacing each occurrence of the label $y$ for the label $x$. The $\monl$ and $\monr$ rules are the natural way to capture monotonicity when nested sequents are represented using labels.

\begin{figure}
{\small
$$
\infer[\ctr]
{\seq{\Rsym, \Gamma, \Gamma'}{\Delta}}
{\seq{\Rsym, \Gamma, \Gamma', \Gamma'}{\Delta}}
\quad
\infer[\ctr]
{\seq{\Rsym, \Gamma}{\Delta, \Delta'}}
{\seq{\Rsym, \Gamma}{\Delta, \Delta', \Delta'}}
\quad
\infer[\ctr]
{\seq{\Rsym, \Rsym', \Gamma}{\Delta}}
{\seq{\Rsym, \Rsym', \Rsym', \Gamma}{\Delta}}
\quad
\infer[\wk]
{\seq{\Rsym, \Gamma}{\Delta, \Delta'}}
{\seq{\Rsym, \Gamma}{\Delta}}
\quad
\infer[\wk]
{\seq{\Rsym, \Gamma, \Gamma'}{\Delta}}
{\seq{\Rsym, \Gamma}{\Delta}}
$$
$$
\infer[\wk]
{\seq{\Rsym, \Rsym', \Gamma}{\Delta}}
{\seq{\Rsym, \Gamma}{\Delta}}
\quad
\infer[\cut_{1}]
{\seq{\Rsym }{\Gamma}}
{\seq{\Rsym}{\Gamma, x:A} & \seq{\Rsym}{\Gamma, x:\overline{A}}}
\quad
\infer[\cut_{2}]
{\seq{\Rsym, \Gamma}{\Delta}}
{\seq{\Rsym, \Gamma}{\Delta, x:A} & \seq{\Rsym, x:A, \Gamma}{\Delta}}
$$}
\caption{Admissible rules.}
\label{fig:admissible-rules-biint}
\end{figure}

\begin{lemma}\label{lm:biint-labelled-calc-properties} The calculus $\biintone$
  enjoys the following: (1) admissibility of $\ctr$, $\wk$ and $\cut_{2}$ from
  Fig.~\ref{fig:admissible-rules-biint};
  (2) invertibility of all inference rules from Fig.~\ref{fig:BiIntL-labelled-rules};
  (3) if $\pseq{\Rsym,Rxy,\Gamma}{\Delta}$, then
  $\pseq{\Rsym[x/y],\Gamma[x/y]}{\Delta[x/y]}$, and
  (4) If $\pseq{\Gamma}{\Delta}$
  where $\Gamma$ and $\Delta$ only contain formulae solely labelled with $y$,
  then $\pseq{\Gamma[x/y]}{\Delta[x/y]}$ for any label $x$.
\end{lemma}

\begin{proof} 
For proofs of (1)-(4), see~\cite[Section~3]{PinUus18}. Statement (4)
follows from the others.
\end{proof}

\begin{figure}
{\small
$$
\infer[id]
{\seq{\Rsym,x:p,\Gamma}{ \Delta, x:p}}
{}
\qquad
\infer[\top]
{\seq{\Rsym,\Gamma}{x:\top, \Delta} }
{}
\qquad
\infer[\bot]
{\seq{\Rsym,\Gamma,x:\bot}{\Delta} }
{}
$$
$$
\infer[\lor_{L}]
{\seq{\Rsym,\Gamma, x : A \lor B}{\Delta}}
{\seq{\Rsym,\Gamma,x:A}{\Delta} & \seq{\Rsym,\Gamma,x:B}{\Delta}}
\qquad
\infer[\lor_{R}]
{\seq{\Rsym,\Gamma}{x:A \lor B, \Delta}}
{\seq{\Rsym, \Gamma}{x:A,x:B, \Delta}}
\qquad
\infer[\monl]   
{\seq{\Rsym,Rxy,x:A,\Gamma}{\Delta} }
{
{\seq{\Rsym,Rxy,x:A,y:A,\Gamma}{\Delta} }
}
$$
$$
\infer[\land_{R}]
{\seq{\Rsym,\Gamma}{x:A \land B, \Delta}}
{\seq{\Rsym, \Gamma}{x:A, \Delta} & \seq{\Rsym, \Gamma}{x:B, \Delta}}
\qquad
\infer[\land_{L}]
{\seq{\Rsym,\Gamma, x : A \land B}{\Delta}}
{\seq{\Rsym,\Gamma,x:A,x:B}{\Delta}}
\qquad
\infer[\monr]   
{\seq{\Rsym,Rxy,\Gamma}{y:A,\Delta} }
{
{\seq{\Rsym,Rxy,\Gamma}{x:A,y:A,\Delta} }
}
$$
$$
\infer[\excl_{L}, \mbox{ $y$ fresh}]
{\seq{\Rsym,x:A \excl B, \Gamma}{\Delta} }
{\seq{\Rsym,Ryx,y:A,\Gamma}{y:B, \Delta} }
\quad
\infer[\supset_{L}]
{\seq{\Rsym, x:A \supset B,\Gamma}{\Delta} }
{
\seq{\Rsym,x:A \supset B,\Gamma}{x:A,\Delta} & \seq{\Rsym,x:B,\Gamma}{\Delta} 
}
$$
$$
\infer[\excl_{R}]
{\seq{\Rsym, \Gamma}{x:A \excl B,\Delta} }
{
\seq{\Rsym, \Gamma}{x:A,\Delta}  & \seq{\Rsym, x:B,\Gamma}{x:A \excl B,\Delta} 
}
\quad
\infer[\supset_{R}, \mbox{ $y$ fresh}]
{\seq{\Rsym, \Gamma}{x:A \supset B, \Delta} }
{\seq{\Rsym,Rxy,\Gamma,y:A}{y:B, \Delta} }
$$

}
\caption{The calculus $\biintone$ for $\biint$~\cite{PinUus18}.}
\label{fig:BiIntL-labelled-rules}
\end{figure}

As in the case with tense logics, we define a generalised interpolant to be a set of  two-sided flat sequents. 
However, to ease the definition of $\trans$, we shall use an encoding of two-sided sequents into single-sided sequents
by annotating the left-hand side occurrence of a formula with a $L$ and the right-hand side occurrence with a $R$.
In this way some results concerning intuitionistic interpolants can be easily adapted from the classical counterparts.

\begin{definition}
A {\em polarised formula} is a  formula annotated with $L$ (left-polarised) or $R$ (right-polarised). 
We write $A^L$ ($A^R$) for the left-polarised (right-polarised) version of formula $A.$
A {\em labelled polarised formula} is a polarised formula further annotated with a label. We write $x:A^L$ ($x:A^R$)
to denote a left-polarised  (a right-polarised) formula labelled with $x.$
Given a polarised formula $A^L$ (resp. $A^R$), its dual is defined as
$\overline{A^L} = A^R$ and $\overline{A^R} = A^L.$
That is, duality changes polarities (the side where the formula occurs), but not the actual
formula.
\end{definition}

\begin{definition}
A {\em polarised (flat) sequent} is a single-sided (flat) sequent where all formulas in the sequent are polarised. 
Given a two-sided sequent $S = \seq{\Rsym, x_1:A_1,\ldots,x_m:A_m}{y_1:B_1,\ldots, y_n:B_n}$,
its corresponding polarised sequent is the following
$$
\seq{\Rsym}{x_1:A_1^L, \ldots, x_m: A_m^L, y_1:B_1^R, \ldots, y_n:B_n^R}. 
$$
Given a two-sided sequent $S$, we denote with $\mathsf{P}(S)$ its encoding as a polarised sequent. Conversely,
given a polarised sequent $S$, we denote with $\mathsf{T}(S)$ its two-sided counterpart. This notation extends
to sets of sequents by applying the encoding element-wise. 
\end{definition}

\begin{definition}
An {\em intuitionistic interpolant} is a set of two-sided flat sequents. Given an intuitionistic interpolant $\Isym$,
its orthogonal $\orth{\Isym}$ is defined as $\mathsf{T}(\orth{\mathsf{P}(\Isym)}).$
\end{definition}
\begin{example}
Let $\Isym = \{ (\seq{x:A}{y:B}), (\seq{}{u:C, v:D})\}$. Then, $\orth{\Isym}$ is the set:
$$
\{(\seq{u:C}{x:A}), (\seq{v:D}{x:A}), (\seq{u:C,y:B}{}), (\seq{v:D,y:B}{}) \}
$$
\end{example}

By defining orthogonality via the embedding into polarised sequents,
the Persistence Lemma for the intuitionistic case comes for free, by
appealing to Lem.~\ref{lm:double-orthogonal}. Note that for sequents $\seq{\Gamma_1}{\Delta_1}$ and $\seq{\Gamma_2}{\Delta_2}$, we write
$\seq{\Gamma_1}{\Delta_1} \subseteq \seq{\Gamma_2}{\Delta_2}$ iff $\Gamma_1 \subseteq \Gamma_2$
and $\Delta_1 \subseteq \Delta_2.$

\begin{lemma}[Persistence]
If $\Lambda \in \orth{\orth{\Isym}}$, then there exists a $\Lambda' \in \Isym$ such that 
$\Lambda' \subseteq \Lambda.$
\end{lemma}

\begin{lemma}[Duality]
\label{lm:dual-int}
Given an intuitionistic interpolant $\Isym$, the empty sequent is derivable from $\Isym \cup \orth{\Isym}$ using the $\cut_{2}$ rule and
the contraction rule.
\end{lemma}

\begin{proof} Similar to Lem.~\ref{lm:dual}.
\end{proof}


\begin{definition}
\label{def:imp-excl-interpolants}
Let $\Isym$ be the interpolant below
where $y$ does not occur in $\Gamma_1, \ldots, \Gamma_n, \Delta_1, \ldots,
\Delta_n$:
$$
\begin{array}{ll}
&\{ ( \seq{\Gamma_1,  y:C_{1,1}, \ldots, y:C_{1,k_1}}{\Delta_1, y:D_{1,1}, \ldots, y:D_{1,j_1}} ),  \ldots, \\
& \hspace{.5em}  
( \seq{\Gamma_n, y:C_{n,1}, \ldots, y:C_{n,k_n}}{\Delta_n, y:D_{n,1}, \ldots, y:D_{n,j_n}} ) \}\\
\end{array}
$$

The interpolants $\excl \Isym^x_y$ and $\supset \Isym^x_y$ are shown below
where empty conjunction denotes $\top$ and empty disjunction
denotes $\bot$:
\[
\begin{aligned}
\excl \Isym^x_y & = &\{  (\seq{\Gamma_1,  x: \bigwedge^{k_1}_{i=1} C_{1,i} \excl \bigvee^{j_1}_{i=1} D_{1,i}}{\Delta_1} ) , \ldots, 
    (\seq{\Gamma_n, x: \bigwedge^{k_n}_{i=1} C_{n,i} \excl \bigvee^{j_n}_{i=1} D_{n,i}}{\Delta_n})  \} \\
\supset \Isym^x_y & = & \{  ( \seq{\Gamma_1}{\Delta_1,  x: \bigwedge^{k_1}_{i=1} C_{1,i} \supset \bigvee^{j_1}_{i=1} D_{1,i}} ) , \ldots, 
   ( \seq{\Gamma_n}{\Delta_n, x:\bigwedge^{k_n}_{i=1} C_{n,i} \supset \bigvee^{j_n}_{i=1} D_{n,i}}  )  \}
\end{aligned}
\]
\end{definition}

The proof system 
$\biinttwo$
for constructing intuitionistic interpolants is given in Fig.~\ref{fig:BiInt-Interpolation-labelled-rules}.


\begin{figure}[t]
{\small
$$
\infer[id]
{\iseq{\Rsym,\Gamma_{1},x:p \mid  \Gamma_{2}}{ \Delta_{1} \mid x:p, \Delta_{2}}{\{ \seq{}{x:p}\}}
}
{}
\qquad
\infer[id]
{\iseq{\Rsym,\Gamma_{1} \mid x:p,  \Gamma_{2}}{ \Delta_{1} \mid x:p, \Delta_{2}}{\{\seq{}{x:\bot}\}}
}
{}
$$
$$
\infer[\top]
{\iseq{\Rsym,\Gamma_{1} \mid \Gamma_{2}}{\Delta_{1} \mid x:\top, \Delta_{2}} {\{\seq{}{x:\bot} \}}
}
{}
\qquad
\infer[\bot]
{\iseq{\Rsym,\Gamma_{1} \mid x:\bot,\Gamma_{2}}{\Delta_{1} \mid \Delta_{2}} {\{\seq{}{x:\bot} \}}
}
{}
$$
$$
\infer[\orthrule]
{\iseq{\Rsym, \Gamma_{2} \mid \Gamma_{1}}{\Delta_{2} \mid \Delta_{1}}{\orth{\Isym}}}
{\iseq{\Rsym, \Gamma_{1} \mid \Gamma_{2}}{\Delta_{1} \mid \Delta_{2}}{\Isym}}
$$
$$
\infer[\monl]
{\iseq{\Rsym,Rxy,\Gamma_{1} \mid x:A,\Gamma_{2}}{\Delta_{1} \mid \Delta_{2}}{\Isym} }
{
{\iseq{\Rsym,Rxy,\Gamma_{1} \mid x:A,y:A,\Gamma_{2}}{\Delta_{1} \mid \Delta_{2}}{\Isym} }
}
\qquad
\infer[\monr]
{
{\iseq{\Rsym,Rxy,\Gamma_{1}\mid \Gamma_{2}}{\Delta_{1} \mid y:A, \Delta_{2}}{\Isym} }
}
{\iseq{\Rsym,Rxy,\Gamma_{1} \mid \Gamma_{2}}{\Delta_{1} \mid x:A,y:A, \Delta_{2}}{\Isym} }
$$
$$
\infer[\lor_{L}]
{\iseq{\Rsym,\Gamma_{1} \mid x : A \lor B, \Gamma_{2}}{\Delta_{1} \mid \Delta_{2}} {\Isym_{1} \cup \Isym_{2} }
}
{\iseq{\Rsym,\Gamma_{1} \mid x:A,\Gamma_{2}}{\Delta_{1} \mid \Delta_{2}}{\Isym_{1}} & \iseq{\Rsym,\Gamma_{1} \mid x:B, \Gamma_{2}}{\Delta_{1} \mid \Delta_{2}}{\Isym_{2} }
}
$$
$$
\infer[\land_{R}]
{\iseq{\Rsym,\Gamma_{1} \mid \Gamma_{2}}{\Delta_{1} \mid x:A \land B, \Delta_{2}} {\Isym_{1} \cup \Isym_{2} }
}
{\iseq{\Rsym,\Gamma_{1} \mid \Gamma_{2}}{\Delta_{1} \mid x:A, \Delta_{2}}{\Isym_{1}} & \iseq{\Rsym,\Gamma_{1} \mid \Gamma_{2}}{\Delta_{1} \mid x:B, \Delta_{2}}{\Isym_{2} }
}
$$
$$
\infer[\land_{L}]
{\iseq{\Rsym,\Gamma_{1} \mid x:A \land B,\Gamma_{2}}{\Delta_{1} \mid \Delta_{2}}{\Isym}
}
{\iseq{\Rsym,\Gamma_{1} \mid x:A, x:B, \Gamma_{2}}{\Delta_{1} \mid \Delta_{2}} {\Isym }
}
\qquad
\infer[\lor_{R}]
{\iseq{\Rsym,\Gamma_{1} \mid \Gamma_{2}}{\Delta_{1} \mid x:A \lor B, \Delta_{2}}{\Isym}
}
{\iseq{\Rsym,\Gamma_{1} \mid \Gamma_{2}}{\Delta_{1} \mid x:A,x:B, \Delta_{2}} {\Isym }
}
$$
$$
\infer[\supset_{L}]
{\iseq{\Rsym, \Gamma_{1} \mid x:A \supset B, \Gamma_{2}}{\Delta_{1} \mid \Delta_{2}} {\Isym_{1} \cup \Isym_{2}} }
{
\iseq{\Rsym,\Gamma_{1} \mid x:A \supset B, \Gamma_{2}}{\Delta_{1} \mid x:A, \Delta_{2}}{\Isym_{2}} & \iseq{\Rsym,\Gamma_{1} \mid x:B, \Gamma_{1}}{\Delta_{1} \mid \Delta_{2}}{\Isym_{1}}
}
$$
$$
\infer[\excl_{R}]
{\iseq{\Rsym, \Gamma_{1} \mid \Gamma_{2}}{\Delta_{1}  \mid x:A \excl B, \Delta_{2}} {\Isym_{1} \cup \Isym_{2}}
}
{
\iseq{\Rsym, \Gamma_{1} \mid \Gamma_{2}}{\Delta_{1} \mid x:A, \Delta_{2}}{\Isym_{1}}  & \iseq{\Rsym, \Gamma_{1} \mid x:B, \Gamma_{2}}{\Delta_{1}, x:A \excl B \mid \Delta_{2}}{\Isym_{2}} 
}
$$
$$
\infer[\excl_{L}]
{\iseq{\Rsym,\Gamma_{1} \mid x:A \excl B, \Gamma_{2}}{\Delta_{1} \mid \Delta_{2}}{\excl \Isym^y_x} }
{\iseq{\Rsym,Ryx,\Gamma_{1} \mid y:A,\Gamma_{2}}{\Delta_{1} \mid y:B, \Delta_{2}} {\Isym} }
\qquad
\infer[\supset_{R}]
{\iseq{\Rsym, \Gamma_{1} \mid \Gamma_{2}}{\Delta_{1} \mid x:A \supset B, \Delta_{2}} {\supset \Isym^y_x} }
{\iseq{\Rsym,Rxy,\Gamma_{1} \mid y:A, \Gamma_{2}}{\Delta_{1} \mid y:B, \Delta_{2}} {\Isym}}
$$
}
\caption{The calculus $\biinttwo$ used to compute interpolants for $\biint$. In $\supset_R$ and $\excl_L$, $y$ is fresh.}
\label{fig:BiInt-Interpolation-labelled-rules}
\end{figure}

\begin{lemma}
\label{lm:int-interpolant-propvar-labels}
If $\pseq{\Rsym, \Gamma_{1}, \Gamma_{2}}{\Delta_{1}, \Delta_{2}}$, then there exists an $\Isym$ 
such that
$\piseq{\Rsym, \Gamma_{1} \mid \Gamma_{2}}{\Delta_{1} \mid \Delta_{2}}{ \Isym}$, $var(\Isym) \subseteq var(\Gamma_{1}, \Delta_{1}) \cap var(\Gamma_{2}, \Delta_{2})$, and all labels in $\Isym$ also occur in $\Rsym, \Gamma_{1}, \Delta_{1}$ or $\Gamma_{2}, \Delta_{2}$.
\end{lemma}
\begin{proof}
  Induction on the height of the proof of
  $\seq{\Rsym, \Gamma_{1}, \Gamma_{2}}{\Delta_{1}, \Delta_{2}}$ using 
  rules of $\biinttwo$.
\end{proof}

The main technical lemma below asserts that the interpolants constructed via 
the proof system in Fig.~\ref{fig:BiInt-Interpolation-labelled-rules} 
obey duality properties which are essential for proving the main
theorem (Thm.~\ref{tm:biint-interpolation}). The proof of this lemma can be
found in App.~\ref{app-proofs}. 

\begin{lemma}
\label{lm:int-interpolant}
For all $\Rsym, \Gamma_1, \Gamma_2, \Delta_1, \Delta_2$ and $\Isym$, 
if $\piseq{\Rsym, \Gamma_1 \mid \Gamma_2}{ \Delta_1 \mid \Delta_2}{\Isym}$, then 
\begin{enumerate}
\item For all $(\seq{\Sigma}{\Theta}) \in \Isym$, we have $\pseq{\Rsym, \Gamma_1, \Sigma}{\Theta, \Delta_1}$ and
\item For all $(\seq{\Lambda}{\Omega}) \in \orth{\Isym}$, we have $\pseq{\Rsym, \Gamma_2, \Lambda}{\Omega, \Delta_2}.$
\end{enumerate}
\end{lemma}

Given a sequent $\Lambda$, we denote with $\Lambda^L$ (resp., $\Lambda^R$) 
the multiset of labelled formulas on the left (resp. right) hand side of $\Lambda$.
The following two lemmas are counterparts of Lem.~\ref{lm:deriving-interpolant-formula} and Lem.~\ref{lm:deriving-dual-interpolant-formula}. 
Lem.~\ref{lm:deriving-biint-interpolant-formula} essentially states that in a specific case, an interpolant
can be interpreted straightforwardly as a conjunction of implications. Its proof is given in App.~\ref{app-proofs}. 
The proof of Lem.~\ref{lm:deriving-biint-dual-interpolant-formula} follows the same pattern as in the proof of Lem.~\ref{lm:deriving-dual-interpolant-formula}.

\begin{lemma}
\label{lm:deriving-biint-interpolant-formula}
Let $\Isym = \{(\seq{\Sigma_{1}}{\Theta_{1}}), \ldots, (\seq{\Sigma_{n}}{\Theta_{n}})\}$ be an interpolant with 
$$
(\seq{\Sigma_{i}}{\Theta_{i}}) = (\seq{x:C_{i,1}, \ldots, x:C_{i,k_{i}}}{ x:D_{i,1}, \ldots, x:D_{i,j_{i}}})
\mbox{ for each } 1 \leq i \leq n.
$$
If $\pseq{\Sigma_{i},\Gamma}{\Theta_{i}}$, for all $(\seq{\Sigma_{i}}{\Theta_{i}}) \in \Isym$, and every formula in $\Gamma$ is labelled with $x$, then 
$$
\pseq{\Gamma}{x:\bigwedge_{i=1}^{n} (\bigwedge_{m=1}^{k_{i}} C_{i,m} \supset \bigvee_{m=1}^{j_{i}} D_{i,m}}).
$$
\end{lemma}

\begin{lemma}
\label{lm:deriving-biint-dual-interpolant-formula}
Let $\Isym = \{(\seq{\Sigma_{1}}{\Theta_{1}}), \ldots, (\seq{\Sigma_{n}}{\Theta_{n}})\}$ be an interpolant with 
$$
(\seq{\Sigma_{i}}{\Theta_{i}}) = (\seq{x:C_{i,1}, \ldots, x:C_{i,k_{i}}}{ x:D_{i,1}, \ldots, x:D_{i,j_{i}}})
\mbox{ for each } 1 \leq i \leq n.
$$
If $\pseq{\Rsym, \Sigma_{i},\Gamma}{\Delta, \Theta_{i}}$ for all $(\seq{\Sigma_{i}}{\Theta_{i}}) \in \orth{\Isym}$, then 
$$
\pseq{\Rsym, \Gamma,  x:\bigwedge_{i=1}^{n} (\bigwedge_{m=1}^{k_{i}} C_{i,m} \supset \bigvee_{m=1}^{j_{i}} D_{i,m})}{\Delta.}
$$
\end{lemma}

\begin{proof} Follows from Lem.~\ref{lm:dual-int} and is similar to Lem.~\ref{lm:deriving-dual-interpolant-formula}.

\end{proof}

\begin{theorem}
\label{tm:biint-interpolation}
If $\pseq{}{x:A \supset B}$, then there exists a $C$ such that (i) $var(C) \subseteq var(A) \cap var(B)$ and (ii) $\pseq{}{x:A \supset C}$ and $\pseq{}{x:C \supset B}$.
\end{theorem}

\begin{corollary}
\label{cr:biint-lydon-craig-interpolation}
The logic $\biint$ has the Craig interpolation property.
\end{corollary}

\section{Conclusion and Future work}
\label{sec:concl}

We have presented a novel approach to proving the interpolation theorem for a range of logics possessing a nested sequent calculus. The key insight in our approach is the generalisation of the interpolation theorem to allow sets of sequents as interpolants. There is a natural definition of duality between interpolants via cut. We have shown that our method can be used to prove interpolation for logics for which interpolation was known to be difficult to prove.\looseness=-1

We intend on applying our approach to prove interpolation for bi-intuitionistic linear logic (BiILL)~\cite{CloustonDGT13}. Unlike tense logics and bi-intuitionistic logic, there is no obvious Kripke semantics for BiILL, so Kuznets et. al.'s approach is not immediately applicable, and it seems a proof-theoretic approach like ours would offer some advantage. 
We conjecture that the key insight in our work, i.e., the generalisation of interpolants to sets of sequents and the use of orthogonality to define duality between interpolants, can be extended to the linear logic setting; for example, via a similar notion of orthogonality as in multiplicative linear logic ~\cite{DanosR89}.

\bibliography{interpolation}

\newpage

\appendix

\section{Definition of $\Rsym^{{\sf\Pi}}$}
In this appendix we define the relation $\Rsym^{{\sf \Pi}}$ that is used as a side condition in the $\FDia$ and $\PDia$ rules in $\tenseone$. Concepts needed for the definition are defined first followed by the definition of the relation.

\begin{definition}[Path Axiom Inverse~\cite{GorPosTiu11}] Let $\QDia^{-1} = \FDia$ if $\QDia = \PDia$, and $\QDia^{-1} = \PDia$, if $\QDia = \FDia$. If $F$ is a path axiom of the form below left, then we define the \emph{inverse of $F$} (denoted $I(F)$) to be the axiom below right:
$$
\QDia_{F_1} \cdots \QDia_{F_n} p \rightarrow \QDia_{F} p \qquad
I(F) = \QDia^{-1}_{F_n} \cdots \QDia^{-1}_{F_1} p \rightarrow \QDia^{-1}_{F} p
$$

\noindent
Given a set of path axioms ${\sf\Pi}$, we define the \emph{set of inverses} to be the set $I({\sf\Pi}) = \{I(F) | F \in {\sf\Pi}\}$.

\end{definition}

\begin{definition}[Composition of Path Axioms~\cite{GorPosTiu11}] Given two path axioms
$$
{F = \QDia_{F_1} \cdots \QDia_{F_n} p \rightarrow \QDia_{F} p \qquad G = \QDia_{G_1} \cdots \QDia_{G_m} p \rightarrow \QDia_{G} p}
$$

\noindent
we say  \emph{$F$ is composable with $G$ at $i$} iff $\QDia_{F} = \QDia_{G_{i}}$.

We define the \emph{composition}
$$
F \triangleright^{i} G = \QDia_{G_1}\cdots\QDia_{G_{i-1}} \QDia_{F_1} \cdots \QDia_{F_n} \QDia_{G_{i+1}} \cdots \QDia_{G_m} p \rightarrow \QDia_{G} p
$$

\noindent
when $F$ is composable with $G$ at $i$.

Using these individual compositions, we define the following \emph{set of compositions}:
$$
F \triangleright G = \{F \triangleright^{i} G \ | \ \text{F is composable with G at $i$} \}
$$

\end{definition}

\begin{definition}[Completion~\cite{GorPosTiu11}] The \emph{completion} of a set ${\sf\Pi}$ of path axioms, written ${\sf\Pi}^{*}$, is the smallest set of path axioms containing ${\sf\Pi}$ such that:

\begin{itemize}

\item $\FDia p \rightarrow \FDia p, \PDia p \rightarrow \PDia p \in {\sf\Pi}^{*}$;

\item If $F,G \in {\sf\Pi}^{*}$ and $F$ is composable with $G$ for some $i$, then $F \triangleright G \subseteq {\sf\Pi}^{*}$.

\end{itemize}

\end{definition}

\begin{definition}[Propagation Graph] 
Let $\seq{\Rsym}{\Gamma}$ be a labelled sequent where $N$ is the set of labels occurring in the sequent. We define the \emph{propagation graph $PG(\seq{\Rsym}{\Gamma}) = ( N,E,L )$} to be the directed graph with the set of nodes $N$, and where the set of edges $E$ and function $L$ that labels edges with either a $\FDia$ or $\PDia$ are defined as follows: For every $Rxy \in \Rsym$, we have an edge $(x,y) \in E$ with $L(x,y) = \FDia$ and an edge $(y,x) \in E$ with $L(y,x) = \PDia$.

\end{definition}

\begin{definition}[Path~\cite{GorPosTiu11}] A \emph{path} is a sequence of nodes and diamonds (labelling edges) of the form:

$$
x_{1},\QDia_{1},x_{2},\QDia_{2},\cdots,\QDia_{n-1},x_{n}
$$

\noindent
in the propagation graph $PG(\Gamma)$ of a sequent $\Gamma$ such that $x_{i}$ is connected to $x_{i+1}$ by an edge labelled with $\QDia_{i}$. For a given path $\pi = x_{1},\QDia_{1},x_{2},\QDia_{2},\cdots,\QDia_{n-1},x_{n}$, we define the \emph{string of $\pi$} to be the string of diamonds $s(\pi) = \QDia_{1} \QDia_{2} \cdots \QDia_{n-1}$.

\end{definition}

\begin{definition}[The Relation $\Rsym^{{\sf\Pi}}$] Let ${\sf\Pi}$ be a set of path axioms. For any two labels $x$ and $y$ occurring in a labelled sequent $\seq{\Rsym}{\Gamma}$, $x\Rsym^{{\sf\Pi}}y$ holds if and only if there exists a path $\pi$ in $PG(\seq{\Rsym}{\Gamma})$ such that $s(\pi) p \rightarrow \QDia p \in {\sf\Pi}^{*}$.

\end{definition}

\label{app:path-axioms-relation}

\section{Proofs}
\begin{customlem}{\ref{lm:dual}}
Given an interpolant $\Isym$, the empty sequent is derivable from $\Isym \cup \orth{\Isym}$ using the cut rule and
the contraction rule.
\end{customlem}
\begin{proof}
By induction on the size of $\Isym.$
Suppose 
$$
\Isym = \{ (\seq{}{x_1:A_1, \ldots, x_n:A_n}) \} \cup \Isym'.
$$
Suppose 
$$
\orth{\Isym'} = \{ \seq{}{\Lambda_1}, \ldots, \seq{}{\Lambda_k} \}.
$$
Then by the definition of orthogonal transformation:
$$
\orth{\Isym} = 
\bigcup_{i=1}^k \bigcup_{j=1}^n \{ (\seq{}{\Lambda_i, x_{j}:\overline{A_j} })  \}. 
$$
So  each $\Lambda_i \in \orth{\Isym'}$ corresponds to exactly $n$ sequents in $\orth{\Isym}$: 
$$
\seq{}{x_{1}:\overline{A_1}, \Lambda_i} \qquad \cdots \qquad \seq{}{x_{n}:\overline{A_n}, \Lambda_i}.
$$
From these, we can construct a derivaton $\Xi_i$ of $\seq{}{\Lambda_i}$ using only cut and contraction:
$$
\deduce
{\seq{}{\Lambda_i}}
{
\deduce{\vdots}{\seq{}{x_1:A_1, \ldots, x_n:A_n} \qquad \seq{}{x_1:\overline{A_1}, \Lambda_i} \quad \cdots \quad \seq{}{x_n:\overline{A_n}, \Lambda_i}}
}
$$
By induction hypothesis, we have a derivation of $\seq{}{}$ from $\Isym' \cup \orth{\Isym'}$:
$$
\deduce{\seq{}{}}
{\deduce{\vdots}{\Isym' \qquad \seq{}{\Lambda_1} \cdots \seq{}{\Lambda_k}}}
$$
By replacing every leaf $\seq{}{\Lambda_i}$ with derivation $\Xi_i$, 
we get a derivation of the empty sequent from assumptions $\Isym$ and $\orth{\Isym'}$, using only cut and contraction as required.
\end{proof}

\begin{customlem}{\ref{lm:interpolant}}
For all $\Rsym, \Gamma, \Delta$ and $\Isym$, 
if $\piseq{\Rsym}{ \Gamma \mid \Delta}{\Isym}$, then 
\begin{enumerate}
\item For all $(\seq{}{\Lambda}) \in \Isym,$ we have $\pseq{\Rsym}{ \Gamma, \Lambda}$ and 
\item For all $(\seq{}{\Theta}) \in \orth{\Isym}$, we have $\pseq{\Rsym}{ \Theta, \Delta}.$
\end{enumerate}
\end{customlem}

\begin{proof}
By  induction on the height of the proof of $\iseq{\Rsym}{ \Gamma \mid \Delta}{\Isym}$; we prove the result for a representative set of cases since all others are simple or shown similarly. The base cases are trivial, so we focus solely on the inductive step.

\paragraph*{\bf $\orthrule$-rule} Suppose the proof of $\iseq{\Rsym}{\Gamma \mid \Delta}{\Isym}$ ends with the $\orthrule$-rule:
$$
\infer[\orthrule]
{\iseq{\Rsym}{\Delta \mid \Gamma}{\orth{\Isym}}}
{\iseq{\Rsym}{\Gamma \mid \Delta}{\Isym}}
$$
\begin{enumerate}
\item Let $\seq{}{\Lambda} \in \orth{\Isym}$: we need to establish that $\pseq{\Rsym}{\Lambda, \Delta}$.
This follows immediately from the induction hypothesis.

\item Let $\seq{}{\Lambda} \in \orth{\orth{\Isym}}$: We need to establish that $\pseq{\Rsym}{\Gamma,\Lambda}.$ By Lem.~\ref{lm:double-orthogonal}, there exists a $\Lambda' \in \Isym$ such that $\Lambda' \subseteq \Lambda$.
By the IH, we have $\pseq{\Rsym}{\Gamma,\Lambda'}$, and applying admissibility of weakening (Lem.~\ref{lm:labelled-calc-properties}),
we get $\pseq{\Rsym}{\Gamma, \Lambda}$ as required.

\end{enumerate}

\paragraph*{\bf $\land$-rule} Suppose 
our proof ends with the inference:
$$
\infer[\land]
{\iseq{\Rsym}{ \Gamma \mid x:A \land B,  \Delta}{\Isym_1 \cup \Isym_2}}
{\iseq{\Rsym}{ \Gamma \mid  x:A, \Delta}{\Isym_1}  & \iseq{\Rsym} {\Gamma \mid x:B,  \Delta}{\Isym_2}}
$$

\begin{enumerate}

\item Let $\seq{}{\Lambda} \in \Isym_{1} \cup \Isym_{2}$. 
By the induction hypothesis, for any  $\seq{}{\Lambda_{1}} \in \Isym_{1}$ and 
$\seq{}{\Lambda_{2}} \in \Isym_{2}$, 
we have that $\pseq{\Rsym}{\Gamma, \Lambda_{1}}$, and 
$\pseq{\Rsym}{\Gamma, \Lambda_{2}}$. Regardless of if $\seq{}{\Lambda} \in \Isym_{1}$ or $\seq{}{\Lambda} \in \Isym_{2}$, we achieve the desired conclusion.

\item Let $\seq{}{\Lambda} \in \orth{\Isym_{1} \cup \Isym_{2}}$.
By the induction hypothesis, for any $\seq{}{\Lambda_{1}} \in \orth{\Isym_{1}}$ and
 $\seq{}{\Lambda_{2}} \in \orth{\Isym_{2}}$, 
 $\pseq{\Rsym}{ x:A, \Delta, \Lambda_{1}}$ and $\pseq{\Rsym} { x:B, \Delta, \Lambda_{2}}$. Observe that by the definition of $\trans$, there is a $\seq{}{\Lambda_{1}} \in \orth{\Isym_{1}}$ and $\seq{}{\Lambda_{2}} \in \orth{\Isym_{2}}$ 
 such that $(\seq{}{\Lambda}) = (\seq{}{\Lambda_{1}, \Lambda_{2}})$. Hence, by admissibility of weakening (Lem.~\ref{lm:labelled-calc-properties}) we can derive ${\seq{\Rsym}{x:A, \Delta, \Lambda_{1}, \Lambda_{2}}}$ and 
 ${\seq{\Rsym}{x:B, \Delta, \Lambda_{1}, \Lambda_{2}}}$, from which, an application of the conjunction rule gives the desired conclusion.

\end{enumerate}

\paragraph*{\bf $\FBox$-rule} Suppose the proof of $\iseq{\Rsym}{\Gamma \mid \Delta}{\Isym}$ ends with a $\FBox$ rule, i.e., $\Delta = \Delta' \cup \{x:\FBox A\}$, $\Isym = \FBox \Isym'^y_x$, and the proof has the form:
$$
\infer[\FBox]
{\iseq{\Rsym}{\Gamma \mid x:\FBox A, \Delta'}{\FBox \Isym'^y_x}}
{\iseq{\Rsym, Rxy}{\Gamma \mid y:A, \Delta'}{\Isym'} }
$$
\begin{enumerate}
\item Let $\seq{}{\Lambda} \in \FBox \Isym'^y_x$: we need to show that
$\pseq{\Rsym}{\Gamma,\Lambda}$. From the definition of $\FBox\Isym'^y_x$, $\Lambda$ must be of
the form $\{\Lambda', x: \FBox (B_1 \lor \cdots \lor B_k)\}$
where $(\seq{}{\Lambda', y:B_1, \ldots, y:B_k}) \in \Isym'.$
The proof of $\seq{\Rsym}{\Gamma, \Lambda}$ is constructed as follows:
$$
\infer[\FBox]
{\seq{\Rsym}{\Gamma, \Lambda', x: \FBox (B_1 \lor \cdots \lor B_k)}}
{
\infer[\lor \times (k-1)]
{\seq{\Rsym, Rxy}{\Gamma, \Lambda', y:B_1 \lor \cdots \lor B_k}}
{\seq{\Rsym, Rxy}{\Gamma, \Lambda', y:B_1, \ldots, y:B_k}}
}
$$
with the premise derivable by the induction hypothesis.

\item Let $\seq{}{\Lambda} \in \orth{\FBox \Isym'^y_x}$: we need to show
that $\pseq{\Rsym}{\Lambda, x:\FBox A, \Delta'}$. In this case, $\Lambda$ could contain zero or more formulae of the form
$x : \overline{\FBox (B_1 \lor \cdots \lor B_k)}$ such that
there exists $\seq{}{\Lambda'} \in \Isym'$ with $\{y:B_1, \ldots, y:B_k\} \subseteq \Lambda'.$ Let us refer to these formulae as a boxed-interpolant formulae. We suppose for the sake of simplicity that exactly one boxed-interpolant formula exists in $\Lambda$ and that $k = 2$; the general case is obtained similarly. We may therefore write $\Lambda$ as 
$\{{x:\overline{\Box(B_{1} \lor B_{2})}} \} \cup \Theta$.

It follows from our assumptions then that there exists $\Lambda' \in \Isym'$ with $\{y:B_1, y:B_2\} \subseteq \Lambda'$. Hence, there must exist a $\Lambda_{1} \in \orth{\Isym'}$ of the form $\{y:\overline{B_{1}}\} \cup \Theta$ and a $\Lambda_{2} \in \orth{\Isym'}$ of the form $\{y:\overline{B_{2}}\} \cup \Theta$. By the induction hypothesis, $\pseq{\Rsym, Rxy}{y:\overline{B}_{1}, {\Theta}, y:A, \Delta'}$ and $\pseq{\Rsym, Rxy}{y:\overline{B}_{2}, {\Theta}, y:A, \Delta'}$. Evoking weakening admissibility from Lem.~\ref{lm:labelled-calc-properties}, we have that
$$
\begin{aligned}
& \pseq{\Rsym, Rxy}{y:\overline{B}_{1}, {\Theta},x:\FDia  (\overline{B}_1 \land \overline{B}_2), y:A, \Delta'}\\
& \pseq{\Rsym, Rxy}{y:\overline{B}_{2}, {\Theta},x:\FDia  (\overline{B}_1 \land \overline{B}_2),y:A, \Delta'}
\end{aligned}
$$
By using the $\land$ rule with $y:B_1$ and $y:B_2$ principal, we can derive $$\seq{\Rsym, Rxy}{y:\overline{B}_{1} \land \overline{B}_2, {\Theta},x:\FDia (\overline{B}_1 \land \overline{B}_2),y:A, \Delta'}.$$ The desired result is obtained as follows:
$$
\infer[\FBox]
{\seq{\Rsym}{{\Theta},x:\FDia  (\overline{B}_1 \land \overline{B}_2),y:\FBox A, \Delta'}}
{
\infer[\FDia]
{\seq{\Rsym, Rxy}{{\Theta},x:\FDia  (\overline{B}_1 \land \overline{B}_2),y:A, \Delta'}}
{\seq{\Rsym, Rxy}{y:\overline{B}_{1} \land \overline{B}_2, {\Theta},x:\FDia  (\overline{B}_1 \land \overline{B}_2),y:A, \Delta'}}
}
$$
\end{enumerate}
\end{proof}

\begin{customthm}{\ref{tm:tense-proof-theoretic-interpolation}}
If $\pseq{}{x:A \rightarrow B}$, then there exists a $C$ such that (i) $var(C) \subseteq var(A) \cap var(B)$ and (ii) $\pseq{}{x:A \rightarrow C}$ and $\pseq{}{x:C \rightarrow B}$.
\end{customthm}

\begin{proof} Assume that $\pseq{}{x:A \rightarrow B}$, \textit{i.e.} $\pseq{}{x:\overline{A} \vee B}$. By Lem.~\ref{lm:labelled-calc-properties}, we know that $\pseq{}{x:\overline{A}, x:B}$. Therefore, by Lem.~\ref{lm:interpolant-labels-literals}, we know that there exists an interpolant 
$\Isym = \{\seq{}{\Lambda_{1}}, \ldots, \seq{}{\Lambda_{n}}\}$ 
with $\Lambda_{i} = \{x:A_{i,1}, \ldots, x:A_{i,k_{i}}\}$ for each $1 \leq i \leq n$ such that $\piseq{}{x:\overline{A} \mid x:B}{\Isym}$, $var(\Isym) \subseteq var(x:A) \cap var(x:B)$, and all formulae in $\Isym$ are labelled with $x$. By Lem.~\ref{lm:interpolant}, the following two statements hold: (a) for all $(\vdash \Lambda) \in \Isym$, $\pseq{}{x:\overline{A},\Lambda}$, (b) for all $(\vdash \Theta) \in \orth{\Isym}$, $\pseq{}{\Theta,x:B}$. We now use the interpolant $\Isym$ to construct a $C$ satisfying the conclusion of the theorem.

The fact that all labelled formulae in $\Isym$ have the same label $x$, along with statement (a) and Lem.~\ref{lm:deriving-interpolant-formula}, imply that $\pseq{}{x:\overline{A},x:\bigwedge \bigvee \Isym}$. Moreover, statement (b) and Lem.~\ref{lm:deriving-dual-interpolant-formula} imply that 
$\pseq{}{x:\overline{\bigwedge \bigvee \Isym},x:B}$. 
Applying the $\lor$ rule in both cases implies that $\pseq{}{x:\overline{A} \vee \bigwedge \bigvee \Isym}$ and $\pseq{}{x:\overline{\bigwedge \bigvee \Isym} \vee B}$, thus giving the desired result with $\bigwedge \bigvee \Isym$ the interpolant. Last, note that since $\bigwedge \bigvee \Isym$ is our interpolant, and $var(\Isym) \subseteq var(x:A) \cap var(x:B)$ holds, $var(\bigwedge \bigvee \Isym) \subseteq var(A) \cap var(B)$ will be satisfied as well. 
\end{proof}



\begin{lemma}
\label{lm:admissibility-left-implication-right-exclusion}
The following two rules are derivable in $\biintone$:
$$
\infer[\supset^{*}_{L}]
{\seq{\Rsym,Rxy,\Gamma,x:A \supset B}{\Delta}}
{\seq{\Rsym,Rxy,\Gamma,x:A \supset B}{\Delta,y:A} & \seq{\Rsym,Rxy,\Gamma,y:B}{\Delta}}
$$
$$
\infer[\excl^{*}_{R}]
{\seq{\Rsym,Ryx,\Gamma}{x:A \excl B,\Delta}}
{\seq{\Rsym,Ryx,\Gamma}{\Delta,y:A} & \seq{\Rsym,Rxy,\Gamma,y:B}{x:A \excl B,\Delta}}
$$
\end{lemma}
\begin{proof} We prove the claim for the $\supset^{*}_{L}$ rule; the proof for $\excl^{*}_{R}$ is similar. Assume that $\pseq{\Rsym,Rxy,\Gamma,x:A \supset B}{\Delta,y:A}$ and $\pseq{\Rsym,Rxy,\Gamma,y:B}{\Delta}$. By evoking admissibility of weakening (Lem.~\ref{lm:biint-labelled-calc-properties}), we obtain proofs of the following sequents:
$$
\begin{aligned}
& \seq{\Rsym,Rxy,\Gamma,x:A \supset B,y:A \supset B}{\Delta,y:A} \qquad
& \seq{\Rsym,Rxy,\Gamma,x:A \supset B,y:B}{\Delta}
\end{aligned}
$$
By applying $\supset_{L}$ with the above sequents, we obtain the premise below; the desired conclusion is derived with one application of $\monl$:
$$
\infer[\monl]
{\seq{\Rsym,Rxy,\Gamma,x:A \supset B}{\Delta}}
{
\seq{\Rsym,Rxy,\Gamma,x:A \supset B,y:A \supset B}{\Delta}
}
$$
\end{proof}

\begin{lemma}\label{lm:admissibility-refl} The rule
$$
\infer[ref]
{\Rsym, \Gamma \vdash \Delta}
{\Rsym, Rxx, \Gamma \vdash \Delta}
$$
is admissible in $\biintone$.
\end{lemma}
\begin{proof} We prove the result by induction on the height of the derivation of $\Rsym, Rxx, \Gamma \vdash \Delta$. The base case is simple since if $\Rsym, Rxx, \Gamma \vdash \Delta$ is an initial sequent, then so is $\Rsym, \Gamma \vdash \Delta$. For all rules, with the exception of the $\monl$ and $\monr$ rule, the inductive step is proven by applying the inductive hypothesis followed by an application of the rule. We therefore focus on the above two cases:

\paragraph*{\bf $\monl$-, $\monr$-rules} If $Rxx$ is not principal in a $\monl$ or $\monr$ inference, then the case is resolved by applying the inductive hypothesis followed by an application of the rule. However, if $Rxx$ is principal, then the inferences are of the following forms:
$$
\begin{tabular}{c @{\hskip 2em} c}
\AxiomC{$\Rsym, Rxx, x:A, x:A, \Gamma \vdash \Delta$}
\UnaryInfC{$\Rsym, Rxx, x:A, \Gamma \vdash \Delta$}
\DisplayProof

&

\AxiomC{$\Rsym, Rxx, \Gamma \vdash x:A, x:A,\Delta$}
\UnaryInfC{$\Rsym, Rxx, \Gamma \vdash x:A,\Delta$}
\DisplayProof
\end{tabular}
$$
Applying the inductive hypothesis to the premises, followed by an application of $\ctr$ admissibility (Lem.~\ref{lm:biint-labelled-calc-properties}), gives the desired conclusion.
\end{proof}

\begin{customlem}{\ref{lm:int-interpolant}}
For all $\Rsym, \Gamma_1, \Gamma_2, \Delta_1, \Delta_2$ and $\Isym$, 
if $\piseq{\Rsym, \Gamma_1 \mid \Gamma_2}{ \Delta_1 \mid \Delta_2}{\Isym}$, then 
\begin{enumerate}
\item For all $(\seq{\Sigma}{\Theta}) \in \Isym$, we have $\pseq{\Rsym, \Gamma_1, \Sigma}{\Theta, \Delta_1}$ and
\item For all $(\seq{\Lambda}{\Omega}) \in \orth{\Isym}$, we have $\pseq{\Rsym, \Gamma_2, \Lambda}{\Omega, \Delta_2}.$
\end{enumerate}
\end{customlem}

\begin{proof} By induction on the height of proofs; we prove the result for a representative set of cases since all others are simple or shown similarly. 
Moreover, the base cases are trivial, so we focus solely on the inductive step, where we show the cases for $\lor_{L}$ and $\supset_{R}$ since all other cases are similar or simple. 

\paragraph*{\bf $\lor_{L}$-rule} Suppose that our derivation ends with the upmost $\lor_{L}$ in Fig.~\ref{fig:BiInt-Interpolation-labelled-rules}.



\begin{enumerate}

\item Let $(\seq{\Sigma}{\Theta}) \in \Isym_{1} \cup \Isym_{2}$. By the inductive hypothesis, we have that for all $(\seq{\Sigma}{\Theta}) \in \Isym_{1}$, $\pseq{\Rsym,\Gamma_{1},\Sigma}{\Delta_{1},\Theta}$ and for all $(\seq{\Sigma}{\Theta}) \in \Isym_{2}$, $\pseq{\Rsym,\Gamma_{1},\Sigma}{\Delta_{1},\Theta}$. Therefore, regardless of if $(\seq{\Sigma}{\Theta})$ is in $\Isym_{1}$ or $\Isym_{2}$, we have that $\pseq{\Rsym,\Gamma_{1},\Sigma }{\Delta_{1},\Theta}$.

\item Let $(\seq{\Sigma}{\Theta}) \in \orth{\Isym_{1} \cup \Isym_{2}}$. By the inductive hypothesis, we have that for all $(\seq{\Sigma}{\Theta}) \in \orth{\Isym_{1}}$, $\pseq{\Rsym,x:A, \Gamma_{2},\Sigma}{\Delta_{2},\Theta}$ and for all $(\seq{\Sigma}{\Theta}) \in \orth{\Isym_{2}}$, $\pseq{\Rsym,x:B,\Gamma_{2},\Sigma}{\Delta_{2},\Theta}$. Observe that by the definition of $\trans$, there will be a $(\seq{\Sigma_{1}}{\Theta_{1}})$ and $(\seq{\Sigma_{2}}{\Theta_{2}})$ such that
 $(\seq{\Sigma_{1}, \Sigma_2}{\Theta_{1}, \Theta_2}) \in \orth{\Isym_{1} \cup \Isym_{2}}$. Hence, by the inductive hypothesis and the admissibility of weakening (Lem.~\ref{lm:biint-labelled-calc-properties}) we may derive 
$$
\seq{\Rsym,x:A, \Gamma_{2},\Sigma_{1}, \Sigma_{2}}{\Delta_{2},\Theta_{1},\Theta_{2}}
\quad\mbox{and}\quad
\pseq{\Rsym,x:B,\Gamma_{2}, \Sigma_{1}, \Sigma_{2}}{\Delta_{2},\Theta_{1},\Theta_{2}}.
$$
 One application of the $\lor_{L}$ rule gives the desired conclusion.

\end{enumerate}

%
%
%

\paragraph*{\bf $\supset_{R}$-rule} Assume that our derivation of $\iseq{\Rsym, \Gamma_{1} \mid \Gamma_{2}}{\Delta_{1} \mid \Delta_{2}} {\Isym}$ ends with a $\supset_{R}$ rule, 
where $\Isym \ = \ \supset {\Isym'}^{y}_{x}$ and the last inference has the form:
$$
\infer[\supset_{R}]
{\iseq{\Rsym, \Gamma_{1} \mid \Gamma_{2}}{\Delta_{1} \mid x:A \supset B,  \Delta_{2}} {\supset \Isym'^{y}_{x}}}
{\iseq{\Rsym,Rxy,\Gamma_{1} \mid y:A , \Gamma_{2}}{\Delta_{1} \mid y:B, \Delta_{2}} {\Isym'} }
$$
\begin{enumerate}

\item Let $(\seq{\Sigma}{\Theta}) \in \Isym$. We want to show that $\pseq{\Rsym, \Gamma_{1}, \Sigma}{\Delta_{1}, \Theta}$. Note that $\seq{\Sigma}{\Theta}$ is 
of the form:
$$
\seq{\Sigma}{x: C_{1} \wedge ... \wedge C_{k} \supset D_{1} \vee ... \vee D_{n}, \Theta'},
$$ 
where $\Theta = x: C_{1} \wedge ... \wedge C_{k} \supset D_{1} \vee ... \vee D_{n}, \Theta'$. This implies that there exists a $(\seq{\Lambda}{\Omega}) \in \Isym'$ of the form 
$$
\seq{\Sigma, y:C_{1},...,y:C_{k}}{y:D_{1},...,y:D_{n}, \Theta'}
$$ 
Therefore, the inductive hypothesis implies that the premise below is derivable:
$$
\infer[\supset_{R}]
{\seq{\Rsym, \Gamma_{1}, \Sigma}{x: C_{1} \wedge ... \wedge C_{k} \supset D_{1} \vee ... \vee D_{n}, \Theta', \Delta_{1}}}
{
 \infer[\land_{L} \times (k-1), \lor_{R} \times (n-1)]
 {\seq{\Rsym, Rxy, \Gamma_{1}, \Sigma, y:C_{1} \wedge ... \wedge C_{k}}{ y: D_{1} \vee ... \vee D_{n}, \Theta', \Delta_{1}}}
 {\seq{\Rsym, Rxy, \Gamma_{1}, \Sigma, y:C_{1}, ..., C_{k}}{ y: D_{1}, ..., y:D_{n}, \Theta', \Delta_{1}}}
}
$$
The last $\supset_{R}$ inference may be applied because, by our assumption, $y$ does not occur in $\Rsym,\Gamma_{1},\Sigma, \Theta', \Delta_{1}$.

\item Let $(\seq{\Sigma}{\Theta}) \in \orth{\Isym}$. Observe that $\Sigma$ will contain zero or more formulae of the form 
$x: C_{1} \land ... \land C_{k} \supset D_{1} \lor ... \lor D_{n}$, which we refer to as an implication-interpolant formula, 
such that there exists a sequent in $\Isym'$ of the form:  
$$
\seq{\Sigma', y:C_{1},...,y:C_{k}}{y:D_{1},...,y:D_{n}, \Theta}. 
$$ 
where $\Sigma'$ is equal to $\Sigma$ minus all implication-interpolant formulae. We assume for the sake of simplicity that $\Sigma$ contains one implication-interpolant formula, that $k=2$, and $n=2$; the general case is tedious, but shown similarly. 
$\seq{\Sigma}{\Theta}$ is therefore of the form: 
$$
\seq{x: C_{1} \land ... \land C_{k} \supset D_{1} \lor ... \lor D_{n}, \Sigma'}{\Theta}
.$$
It follows from our assumptions that there exists a $(\seq{\Lambda}{\Omega}) \in \Isym'$ of the form
$$
\seq{\Sigma', y:C_{1},y:C_{2}}{y:D_{1},y:D_{2}, \Theta}
.$$
Hence, by the definition of the $\trans$, the following four sequents are members of $\orth{\Isym'}$: 
$$
\seq{\Sigma'}{ \Theta, y:C_1} \qquad \seq{\Sigma'}{\Theta, y:C_2} \qquad \seq{y:D_1, \Sigma'}{\Theta} \qquad \seq{y:D_2, \Sigma'}{\Theta}
$$
By the inductive hypothesis the following four claims hold:
$$
\begin{aligned}
& \pseq{\Rsym, Rxy,y:A,\Gamma_2, \Sigma'}{\Theta,\Delta_{2},y:B,y:C_{1}}\\
& \pseq{\Rsym, Rxy,y:A,\Gamma_2, \Sigma'}{\Theta, \Delta_{2},y:B,y:C_{2}}\\
& \pseq{\Rsym, Rxy,y:A,y:D_{1},\Gamma_2, \Sigma'}{\Theta,  \Delta_{2},y:B}\\
& \pseq{\Rsym, Rxy,y:A,y:D_{2},\Gamma_2, \Sigma'}{\Theta, \Delta_{2},y:B}
\end{aligned}
$$
Applying the $\land_{R}$ rule to the first two sequents with $y:C_{1}$ and $y:C_{2}$ active, and the $\lor_{L}$ rule to the latter two sequents with $y:D_{1}$ and $y:D_{2}$ active, gives derivations of the following two sequents:
$$
\begin{aligned}
& \seq{\Rsym, Rxy,y:A,\Gamma_2, \Sigma'}{\Theta,\Delta_{2},y:B,y:C_{1} \land C_2}\\
& \seq{\Rsym, Rxy,y:A,y:D_{1} \lor D_2,\Gamma_2, \Sigma'}{\Theta,  \Delta_{2},y:B}
\end{aligned}
$$
By admissibility of weakening (Lem.~\ref{lm:biint-labelled-calc-properties}), we may weaken in $x:C_{1} \land C_{2} \supset D_{1} \lor D_{2}$ on the left side of both sequents and apply the $\supset^{*}_{L}$ rule (Lem.~\ref{lm:admissibility-left-implication-right-exclusion}) to obtain a derivation of:
$$
\seq{\Rsym, Rxy,y:A, x:C_{1} \land C_{2} \supset D_{1} \lor D_{2}, \Gamma_2, \Sigma'}{\Theta,\Delta_{2},y:B}
$$
An application of $\supset_{R}$ with $Rxy$, $y:A$, and $y:B$ principal gives the desired result.
\end{enumerate}
\end{proof}

\begin{customlem}{\ref{lm:deriving-biint-interpolant-formula}} 
Let $\Isym = \{(\seq{\Sigma_{1}}{\Theta_{1}}), \ldots, (\seq{\Sigma_{n}}{\Theta_{n}})\}$ be an interpolant with 
$$
(\seq{\Sigma_{i}}{\Theta_{i}}) = (\seq{x:C_{i,1}, \ldots, x:C_{i,k_{i}}}{ x:D_{i,1}, \ldots, x:D_{i,j_{i}}})
\mbox{ for each } 1 \leq i \leq n.
$$
If $\pseq{\Sigma_{i},\Gamma}{\Theta_{i}}$, for all $(\seq{\Sigma_{i}}{\Theta_{i}}) \in \Isym$, and every formula in $\Gamma$ is labelled with $x$, then 
$$
\pseq{\Gamma}{x:\bigwedge_{i=1}^{n} (\bigwedge_{m=1}^{k_{i}} C_{i,m} \supset \bigvee_{m=1}^{j_{i}} D_{i,m}}).
$$
\end{customlem}

\begin{proof} By assumption, $\pseq{x:C_{i,1}, \ldots, x:C_{i,k_{i}} ,\Gamma}{x:D_{i,1}, \ldots, x:D_{i,j_{i}}}$ for all $1 \leq i \leq n$. By repeated application of $\land_{L}$ and $\lor_{R}$, we obtain a derivation of the following sequent for each $1 \leq i \leq n$:
$$
\seq{x: \bigwedge_{m=1}^{k_{i}} C_{i,m},\Gamma}{x: \bigvee_{m=1}^{j_{i}} D_{i,m}}
$$
Let $y$ be a fresh variable and let $\Gamma' = \Gamma[y/x]$, i.e., the multiset of labelled formulae $\Gamma$ but with every formula labelled with $y$ instead of $x$. By the admissibility of weakening (Lem.~\ref{lm:biint-labelled-calc-properties}) and the derivability of the above sequents, we may derive the following:
$$
\seq{Ryx, x: \bigwedge_{m=1}^{k_{i}} C_{i,m},\Gamma, \Gamma'}{x: \bigvee_{m=1}^{j_{i}} D_{i,m}}
$$
for all $1 \leq i \leq n$. We construct derivations from the derivations of the above sequents by adding the following inferences to the bottom of each:
$$
\infer[\supset_{R}]
{\seq{\Gamma'}{y: \bigwedge_{m=1}^{k_{i}} C_{i,m} \supset \bigvee_{m=1}^{j_{i}} D_{i,m}}}
{
\infer[\monl]
{\seq{Ryx, x: \bigwedge_{m=1}^{k_{i}} C_{i,m},\Gamma'}{x: \bigvee_{m=1}^{j_{i}} D_{i,m}}}
{\seq{Ryx, x: \bigwedge_{m=1}^{k_{i}} C_{i,m},\Gamma, \Gamma'}{x: \bigvee_{m=1}^{j_{i}} D_{i,m}}}
}
$$
for all $1 \leq i \leq n$. By successively applying $\land_{R}$ we may derive the following:
$$
\seq{\Gamma'}{y:\bigwedge_{i=1}^{n} (\bigwedge_{m=1}^{k_{i}} C_{i,m} \supset \bigvee_{m=1}^{j_{i}} D_{i,m}})
$$
By the fact that the above sequent is derivable and by statement (4) of Lem.~\ref{lm:biint-labelled-calc-properties}, we obtain the desired conclusion:
$$
\pseq{\Gamma}{x:\bigwedge_{i=1}^{n} (\bigwedge_{m=1}^{k_{i}} C_{i,m} \supset \bigvee_{m=1}^{j_{i}} D_{i,m})}
$$
\end{proof}

\begin{customthm}{\ref{tm:biint-interpolation}}
If $\pseq{}{x:A \supset B}$, then there exists a $C$ such that (i) $var(C) \subseteq var(A) \cap var(B)$ and (ii) $\pseq{}{x:A \supset C}$ and $\pseq{}{x:C \supset B}$.
\end{customthm}

\begin{proof} Suppose that $\pseq{}{x:A \supset B}$. By Lem.~\ref{lm:biint-labelled-calc-properties} (2), (4) and Lem.~\ref{lm:admissibility-refl}, this implies that $\pseq{x:A}{x:B}$. By Lem.~\ref{lm:int-interpolant-propvar-labels}, we know there exists an 
$\Isym = \{(\seq{\Sigma_{1}}{\Theta_{1}}), \ldots, (\seq{\Sigma_{n}}{\Theta_{n}})\}$ with 
$$
(\seq{\Sigma_{i}}{\Theta_{i}}) = (\seq{x:C_{i,1}, \ldots, x:C_{i,k_{i}}}{x:D_{i,1}, \ldots, x:D_{i,j_{i}}})
$$ 
for each $1 \leq i \leq n$ such that 
$\piseq{x:A \mid \cdot }{\cdot \mid x:B}{\Isym}$, $var(\Isym) \subseteq var(x:A) \cap var(x:B)$, and $x$ is the only label occurring in $\Isym$. 
By Lem.~\ref{lm:int-interpolant}, the following two statements hold: (a) for all $(\seq{\Sigma}{\Theta}) \in \Isym$, we have 
$\pseq{x:A, \Sigma}{\Theta}$, 
and (b) for all $(\seq{\Sigma}{\Theta}) \in \orth{\Isym}$, we have 
$\pseq{\Sigma}{\Theta, x:B}$.


We now use the interpolant $\Isym$ to construct a formula $C$ satisfying the conditions specified in the claim of the theorem.

By the assumption that $\Isym$ only contains formulae with the label $x$ along with statement (a) and Lem.~\ref{lm:deriving-biint-interpolant-formula}, we obtain the following:
$$
\pseq{x:A}{x:\bigwedge_{i=1}^{n} (\bigwedge_{m=1}^{k_{i}} C_{i,m} \supset \bigvee_{m=1}^{j_{i}} D_{i,m})}
$$
Additionally, statement (b) and Lem.~\ref{lm:deriving-biint-dual-interpolant-formula} imply the following:
$$
\pseq{x:\bigwedge_{i=1}^{n} (\bigwedge_{m=1}^{k_{i}} C_{i,m} \supset \bigvee_{m=1}^{j_{i}} D_{i,m}) }{x:B}
$$
By weakening admissibility (Lem.~\ref{lm:biint-labelled-calc-properties}), we may weaken both sequents with $Ryx$ and apply $\supset_{R}$ to obtain the following:
$$
\pseq{}{y:A \supset \bigwedge_{i=1}^{n} (\bigwedge_{m=1}^{k_{i}} C_{i,m} \supset \bigvee_{m=1}^{j_{i}} D_{i,m})}
$$
$$
\pseq{}{y:\bigwedge_{i=1}^{n} (\bigwedge_{m=1}^{k_{i}} C_{i,m} \supset \bigvee_{m=1}^{j_{i}} D_{i,m}) \supset B }
$$
Evoking statement (4) of Lem.~\ref{lm:biint-labelled-calc-properties} allows us to replace the label $y$ with $x$. Last, since $var(\Isym) \subseteq var(x:A) \cap var(x:B)$, it is easy to see that $\bigwedge_{i=1}^{n} (\bigwedge_{m=1}^{k_{i}} C_{i,m} \supset \bigvee_{m=1}^{j_{i}} D_{i,m})$ is our formula interpolant for $A$ and $B$.
\end{proof}\label{app-proofs}

\end{document}